%% file: main.tex
\documentclass[11pt]{article}

\input{defs.tex}

\newcommand{\deficit}{\mathcal{D}}

\newcommand{\zo}{\{0,1\}}
\newcommand{\define}{\coloneqq}
\newcommand{\unif}{\mathsf{unif}}
\newcommand{\minentropy}{\mathrm{H}_\infty}

\newcommand{\NP}{\ensuremath{\mathsf{NP}}}
\newcommand{\TFNP}{\ensuremath{\mathsf{TFNP}}}
\newcommand{\BPP}{\ensuremath{\mathsf{BPP}}}
\newcommand{\FBPP}{\ensuremath{\mathsf{FBPP}}}

\newcommand{\FPs}{\ensuremath{\mathsf{FPs}}}

\begin{document}
\title{Pseudodeterministic Communication Complexity}

\author{
\begin{tabular}{ccc}
\hspace{10em} & & \hspace{10em} \\[-1em]
Mika G\"o\"os&
Nathaniel Harms&
Artur Riazanov \\[-1mm]
\small\slshape EPFL &
\small\slshape University of British Columbia &
\small\slshape EPFL \\[.5em]  
Anastasia Sofronova &
Dmitry Sokolov &
Weiqiang Yuan\\[-1mm]
\small\slshape EPFL &
\small\slshape EPFL \& Universit{\'{e}} de Montr{\'{e}}al &
\small\slshape EPFL 
\end{tabular} 
}

\date{\today}

\maketitle

\begin{abstract}
    \noindent
    We exhibit an $n$-bit partial function with randomized communication complexity $O(\log n)$ but such that any completion of this function into a total one requires randomized communication complexity $n^{\Omega(1)}$. In particular, this shows an exponential separation between randomized and \emph{pseudodeterministic} communication protocols. Previously, Gavinsky (2025) showed an
    analogous separation in the weaker model of parity decision trees. We use
    lifting techniques to extend his proof idea to communication complexity.
\end{abstract}

\vspace{3em}

\setlength{\cftbeforesecskip}{0pt}
\renewcommand\cftsecfont{\mdseries}
\renewcommand{\cftsecpagefont}{\normalfont}
\renewcommand{\cftsecleader}{\cftdotfill{\cftdotsep}}
\setcounter{tocdepth}{1}
\tableofcontents

\thispagestyle{empty}
\setcounter{page}{0}
\newpage

\section{Introduction}

Here is a simple question about the behavior of
randomized algorithms. A basic statistical task is to distinguish ``few'' vs.~``many''\!,
formalized by the \textsc{Gap Majority} problem:
\begin{align*}
\GapMaj(x) \define 
\begin{cases}
   1 &\text{if}\enspace |x| \geq \frac{2}{3}n, \\
   0 &\text{if}\enspace |x| \leq \frac{1}{3}n, \\
   \ast & \text{otherwise},
\end{cases}
\end{align*}
where $|x|$ denotes the number of 1s in $x \in \zo^n$ and $*$ indicates that we
put no requirement on the output of an algorithm (that is, $\GapMaj$ is a search problem
with 0 and 1 both being acceptable outputs). This problem is difficult to solve
deterministically: It requires $\Omega(n)$ queries for decision trees and parity decision trees,
and $\Omega(n)$ bits of deterministic communication when we turn it into a two-player problem by composing it with an appropriate gadget. For example, composing with XOR yields the 
\textsc{Gap Hamming Distance} problem given by~$\GapHD(x,y) \coloneqq \GapMaj(x
\oplus y)$. On the other hand, for \emph{randomized} decision trees (and the
other models as well), the cost of $\GapMaj$ is only $1$ query because the
algorithm can sample a random coordinate and output it; when $|x| \geq \tfrac23
n$ or $|x| \leq \tfrac13 n$ this is correct with probability at least $2/3$. Our simple question is:
\begin{center}
    \emph{What is the randomized algorithm doing on the $*$ inputs?}
\end{center}
On inputs with $|x|=n/2$, the sampling algorithm will output 0 or 1
with equal probability. But \emph{must it} do this, or can we ask that the
randomized algorithm produce consistent outputs for every input $x$? That is, can we ask that the algorithm computes some \emph{completion} of $\GapMaj$ into a total function: on any input $x$ it should output either value 0 with probability at least~$2/3$,
or value 1 with probability at least $2/3$. Such an algorithm is called
\emph{pseudodeterministic}. Generally, a pseudodeterministic algorithm is a
randomized algorithm that is required to output (with probability at least
$2/3$) a single consistent output for every input; this is desirable not only
for partial boolean functions, but for any problem where there is more than 1
acceptable output for each input, that is, for any search problem or relation. The
goal is to combine the efficiency of randomized algorithms with the consistency
of deterministic algorithms. So, are pseudodeterministic algorithms nearly as
efficient as randomized ones?

We show that, for communication protocols, the answer is \emph{no}: there is a partial boolean function with an efficient randomized protocol but such that every total completion of that function has large randomized complexity.

\begin{theorem}
\label{thm:main}
There is a partial communication problem $\{0,1\}^n\times\{0,1\}^n\to\{0,1,*\}$ with randomized communication complexity~$O(\log n)$
but pseudodeterministic communication complexity $\Omega(\sqrt{n})$.
\end{theorem}

Partial $n$-bit boolean functions with $\poly\log n$ randomized communication
cost are a type of ``\BPP{} search problems'': relations $R \subseteq
\zo^n \times \zo^n \times \zo^*$ computed by an efficient randomized
protocol. By analogy to Turing machine complexity classes, the class of ``$\BPP$
search problems'' would be called $\FBPP$. If we use $\FPs$ to denote the class
of relations admitting pseudodeterministic protocols with cost $\poly\log n$,
\cref{thm:main} is the first explicit example witnessing the separation
\begin{equation}
\label{eq:separation}
    \FPs \subsetneq \FBPP ,
\end{equation}
and it does so in the most restricted setting, where the output is only a single
bit. We remark, however, that \eqref{eq:separation} has a caveat: just like the
Turing machine analogue, the communication class $\FBPP$ has several natural
definitions, which are not
equivalent~\cite{Aaronson10,Goldreich11,Ilango23,Aaronson24}. For one of these
definitions (which requires the output to be efficiently verifiable), the inclusion~\eqref{eq:separation} holds in the \emph{opposite direction} for partial boolean
functions; our \cref{thm:main} gives an explicit witness for
\eqref{eq:separation} under the remaining definitions. Weaker versions of
\eqref{eq:separation}, which allow the outputs to be large and do not require an
explicit example, can be proved by counting; see \cref{sec:easy-separation} for
a discussion of these nuances.

Several recent works \cite{Goldwasser2020,GIPS21,Gavinsky2025,FGHH25,BHHLT25}
asked for lower bounds on pseudodeterministic communication complexity for
$\BPP$ search problems. Many of them proved versions of \eqref{eq:separation} for
weaker models of computation: Goldwasser, Grossman, Mohanty, and
Woodruff \cite{Goldwasser2020} proved it for one-way communication protocols;
Goldwasser, Impagliazzo, Pitassi, and Santhanam \cite{GIPS21} proved it for
decision trees; and Gavinsky~\cite{Gavinsky2025} proved it for parity decision
trees. Our result can thus be viewed as a qualitative strengthening of these
prior works (the upper bound in \cref{thm:main} holds even for randomized
non-adaptive decision trees).

For partial functions, the only known prior separation was due to Blondal,
Hatami, Hatami, Lalov, and Tritiak~\cite{BHHLT25}. They showed an
$O(1)$-vs-$\Omega(\log\log n)$ separation for \textsc{Gap Hamming Distance}, \ie
the XOR-lift $\GapMaj(x \oplus y)$. This would be the ideal function to witness the
separation \eqref{eq:separation} for partial functions, and we suspect it
exhibits the maximum possible separation:
\begin{conjecture} \label{conj:gapHD}
The pseudodeterministic communication complexity of $\GapHD$ is $\Omega(n)$.
\end{conjecture}
\noindent
This was one of the principal inspirations for our present work, but our
\cref{thm:main} ultimately does not prove the separation using $\GapHD$; we
instead use a slightly more complicated lift of $\GapMaj$, as explained next.

\subsection{Our techniques}
\label{sec:techniques}

Our proof extends the work of Gavinsky~\cite{Gavinsky2025}. He proved an
$\Omega(\sqrt n)$ query lower bound for pseudodeterministic parity decision
trees computing the $\GapMaj$ partial function. To turn this function into
a communication problem, we compose it with the \textsc{Inner Product} gadget on
$m=O(\log n)$ bits defined by~$\IP_m(a, b) \define \sum_{i=1}^m a_i b_i
\bmod 2$. This produces a communication problem on~$O(n\log n)$ bits:
\[
    \GapMaj \circ \IP_m^n (x_1, \dotsc, x_n, y_1, \dotsc, y_n)
    \define \GapMaj\left( \IP_m(x_1, y_1), \dotsc, \IP_m(x_n, y_n) \right).
\]
Here $x_i, y_i \in \zo^m$, Alice has all the $x_i$ inputs, and Bob has all the
$y_i$ inputs. This problem has randomized communication complexity $O(\log n)$ because the players can select $O(1)$ random indices $i \in [n]$ and solve $\IP_m(x_i,y_i)$ deterministically.

Standard randomized lifting theorems~\cite{BPPlifting,liftinglowdisc} show that for every partial function~$f\colon\{0,1\}^n\to\{0,1,*\}$ the randomized communication complexity of $f\circ\IP^n_m$ equals roughly the randomized query complexity of~$f$. It is open to prove such a general lifting theorem for pseudodeterministic complexity. The immediate obstacle to applying a randomized lifting theorem to a pseudodeterministic protocol for $f\circ\IP^n_m$ is that the protocol may produce different consistent outputs for inputs~$(x,y)$ and~$(x',y')$ such that $\IP^n_m(x,y)=\IP^n_m(x',y')\in f^{-1}(*)$. Since we cannot invoke known lifting theorems as a black box, we employ a \emph{white-box} approach: We use tools from lifting theory to adapt Gavinsky's argument for communication protocols. This adaptation introduces new technical challenges not present in the setting of parity decision trees. We give an overview of the proof in \cref{sec:overview}, including an exposition of Gavinsky's original argument.

\subsection{History and Motivations}

Pseudodeterministic algorithms were introduced by Gat and Goldwasser~\cite{GG11} (under the name \emph{Bellagio} algorithms), and independently by Huynh and Nordstr\"om~\cite{HN12} in the context of proof complexity (they called them \emph{consistent} algorithms). 
Pseudodeterministic algorithms have applications to cryptography and distributed computing, and have been
studied in several computational models including Turing machines
\cite{GGR13,OS17,LOS21,CLORS23}, decision trees \cite{GGR13,GIPS21,CDM23}, and
parity decision trees \cite{Gavinsky2025}. They are natural and interesting in their
own right, they provide an intermediary between randomized and deterministic
complexity, and they are also related to notions of \emph{replicability} in
machine learning and the natural sciences \cite{GIPS21,ILPS22,CMY23,BHHLT25}.

Aside from intrinsic interest in pseudodeterminism, \cite{Gavinsky2025} points
out its connection to structural questions about randomized communication (which
we discuss in the open problems below), and \cite{GIPS21} argues that
understanding pseudodeterministic communication is an important step towards
understanding the communication complexity of search problems (i.e., relations).
Whereas the communication complexity of \emph{functions} benefits from
connections to well-understood query complexity measures (e.g., sensitivity and
block-sensitivity) via lifting theorems, analogues of these
measures for search problems are not well understood. Pseudodeterministic
algorithms must compute \emph{some} function, and therefore serve as an
intermediate between functions and relations. 

Pseudodeterminism was introduced in \cite{HN12} for applications in proof
complexity. That paper shows that \emph{cutting plane proofs} of CNF
unsatisfiability can be turned into pseudodeterministic communication protocols for the \emph{falsified-clause} search problem where two players search for a clause that is not satisfied by a (distributed)
assignment of variables. Their lower bound for this communication problem
allowed them to prove time--space tradeoffs for cutting planes.

\subsection{Open Problems}
\label{sec:open}

\paragraph*{Improved quantitative bounds.}
Let us start by reiterating that we conjecture \textsc{Gap Hamming Distance} to require pseudodeterministic complexity $\Omega(n)$ (\cref{conj:gapHD}). Proving this seems to require two
steps. First, improve Gavinsky's lower bound of $\Omega(\sqrt n)$ for parity
decision trees computing $\GapMaj$ to $\Omega(n)$. Second, lift this to $\GapHD$. Another strategy may be to directly improve the current lower bound
for the pseudodeterministic communication complexity of $\GapHD$, which is $\Omega(\log\log n)$, and which uses very different techniques \cite{BHHLT25}.

It is important to note that $\textsc{Gap Hamming Distance}$ is, in a precise sense, the \emph{only} partial boolean
communication problem where we can hope to prove an $O(1)$-vs-$\poly(n)$
separation between randomized and pseudodeterministic communication. This is
because \emph{all} $n$-bit partial boolean matrices with randomized cost $O(1)$
are a submatrix of $\GapHD$ on $O(n)$ bits (with some
constant $\alpha < 1/2$ in place of $1/3$ in the gap)~\cite{LS09,FGHH25}. So a
pseudodeterministic lower bound for any of these matrices implies the same lower
bound for $\GapHD$.

\paragraph*{Separation for a $\TFNP$ problem.}
The most outstanding \emph{qualitative} problem that is left open is to prove pseudodeterministic lower bounds for a \emph{total $\NP$ search problem} ($\TFNP$). A relation $R\subseteq\{0,1\}^n\times\{0,1\}^n\times\{0,1\}^*$ is in (communication analogue of) $\TFNP$ if
\begin{itemize}[label=$-$,noitemsep]
\item $R$ is \emph{total}, meaning that for all inputs $(x,y)$ there is a solution $z$ such that $(x,y,z)\in R$; and
\item $R$ is \emph{efficiently checkable}, meaning that there is a deterministic verifier protocol $V$ of cost $\poly\log n$ such that $V(xz,yz)=1$ iff $(x,y,z)\in R$ for all $(x,y,z)$.
\end{itemize}
\begin{conjecture}
\label{conj:tfnp}
There exists some $\TFNP$ communication problem with randomized communication complexity $\poly\log n$ but pseudodeterministic complexity $n^{\Omega(1)}$.
\end{conjecture}

The analogous separation for query complexity was the main result of the paper~\cite{GIPS21}. The communication analogue of $\TFNP$ has been studied explicitly in, e.g.,~\cite{Goos2019,Buss23,Goos2025}, but, implicitly, $\TFNP$ problems are ubiquitous in communication complexity: this includes all (monotone) Karchmer--Wigderson games as well as the aforementioned falsified-clause search problems, which arise when applying communication lower bounds to proof complexity. The survey~\cite{Rezende2022} explains these connections and more: $\TFNP$ problems give a unified lens to study the interconnections between communication protocols, propositional proofs, and boolean circuits.

\paragraph*{Separation for a $\BPP$-verifiable problem.}
In fact, it is still open to prove a weaker result than \cref{conj:tfnp}. If we
define $\BPP$ search problems for communication complexity by adapting the
definition of Goldreich \cite{Goldreich11} for Turing machines, we get the class
of relations $R \subseteq \zo^n \times \zo^n \times \zo^*$ such that
\begin{itemize}[label=$-$,noitemsep]
\item There is a randomized $\poly\log (n)$ cost protocol solving $R$; and
\item There is a randomized $\poly\log (n)$ cost \emph{verifier} protocol $V$:
$\Pr[ V(xz,yz) = 1 ] \geq 2/3$ if $(x,y,z) \in R$ and $\Pr[ V(xz,yz) = 1 ] \leq 1/3$
otherwise, for all $(x,y,z)$.
\end{itemize}
We can then replace the $\TFNP$ requirement in \cref{conj:tfnp} with these
$\BPP$-verifiable problems and the conjecture is still open.

Our result implies this separation for a weaker, more specialized type of
verification. The partial function $\GapMaj \circ \IP$ does \emph{not} satisfy
the above definition, because the verifier cannot exist (we could use it to
solve \emph{exact} majority efficiently, which cannot be done). However, we can
define the relation $R \subseteq \zo^{nm} \times \zo^{nm} \times [n]$ where for
each $x,y$ the valid outputs are the numbers $t \in [n]$ that are within $\pm
n/10$ of the number of 1-valued $\IP_m$ gadgets. This relation can be solved in
cost $O(m) = O(\log n)$ the same way as $\GapMaj \circ \IP_m$. It can also be
\emph{sort-of} verified: given $(x,y,\hat t)$, if the correct number of 1-valued
gadgets is $t$, there is a randomized verifier $V$ which satisfies
\begin{itemize}[label=$-$,noitemsep]
    \item If $|t-\hat t| \leq n/10$ (\ie $(x,y,\hat t) \in R$) then $\Pr[ V(x,y,\hat t) = 1 ] \geq 2/3$; and
    \item If $|t-\hat t| \geq n/9$ then $\Pr[ V(x,y,\hat t) = 1 ] \leq 1/3$.
\end{itemize}
An efficient pseudodeterministic protocol for this estimation problem could be
used as a pseudodeterministic protocol for $\GapMaj \circ \IP_m$, so
\cref{thm:main} shows that it cannot exist. However, the necessary gap between
$n/10$ and $n/9$ means that this does not satisfy the definition of
\cite{Goldreich11}.

\paragraph*{Structure of communication protocols.}
Communication complexity is closely related to the size of monochromatic
rectangles within the communication matrices: efficient deterministic protocols
imply large monochromatic rectangles inside the matrix. The question of whether the same is true for randomized protocols that compute a total boolean function was raised in~\cite{Goos2018}. A striking version of this question is obtained for constant-cost protocols: Chattopadhyay, Lovett, and Vinyals \cite{CLV19}
and Hambardzumyan, Hatami, and Hatami \cite{HHH23} conjecture that
\begin{conjecture}[\cite{CLV19,HHH23}]
\label{conj:large-monochromatic-rectangles}
There exists a function $\eta$ such that every $N \times N$ boolean matrix with
randomized communication cost $c$ has an $\eta(c) \cdot N \times \eta(c) \cdot
N$ monochromatic rectangle.
\end{conjecture}
Completions of the $\GapHD$ matrix \emph{cannot} have such large monochromatic rectangles~\cite{FF81},
and therefore \cref{conj:large-monochromatic-rectangles} already implies lower
bounds on the pseudodeterministic cost of $\GapHD$, as noted in
\cite{FGHH25}. Gavinsky \cite{Gavinsky2025} points out similar implications for parity
decision trees, with monochromatic affine subspaces in place of rectangles. The best progress so far towards \cref{conj:large-monochromatic-rectangles} is to find such monochromatic rectangles in matrices of bounded $\gamma_2$-norm~\cite{Balla2025}.

\section{Warm-Up and Proof Overview}
\label{sec:overview}

Our proof starts with the lower bound for parity decision trees by
Gavinsky~\cite{Gavinsky2025} and upgrades it to a communication lower bound
using techniques from query-to-communication lifting~\cite{GLMWZ16,BPPlifting}. As a
warm-up to our main proof, we will present an exposition of Gavinsky's argument,
but simplified to the special case of decision trees.

Of course, for decision trees (rather than parity decision trees), there is a
simpler proof of a superior $\Theta(n)$ bound: in any completion $f$ of
$\GapMaj$, find the input $x \in f^{-1}(0)$ with largest weight; then, observe
that by fixing the 1-valued coordinates of $x$ and letting the 0-valued
coordinates vary, we obtain the \textsc{Or} problem on $\Omega(n)$ bits, which
implies a lower bound of $\Omega(n)$ queries. However, the more complicated
warm-up proof gives a technique that can be lifted to communication complexity.

\subsection{Warm-Up: Lower Bound for Decision Trees}
\label{section:lb-dt}

\begin{theorem}[Simplified version of \cite{Gavinsky2025}]
    \label{thm:query-main}
    Any completion $f\colon \{0, 1\}^n \rightarrow \{0, 1\}$ of  $\GapMaj$ requires randomized query complexity $\Omega(\sqrt{n})$.
\end{theorem}

\paragraph{Proof sketch.}
The proof has two parts, encapsulated in the \emph{Closeness Lemma}
and the \emph{Stage Lemma}. The Closeness Lemma formalizes the key
observation that a shallow decision tree cannot distinguish between the uniform
distribution over $\{0,1\}^n$ and the distribution where the bits are slightly
biased towards $1$. Specifically, suppose $\bm u \sim \{0,1\}^n$ is uniform random and $\bm x$ is
obtained from $\bm u$ by setting to $1$ a uniformly random set of $\sqrt{n}$
coordinates; that is, we ``sprinkle'' some $\sqrt{n}$ many $1$s into $\bm u$. Then the probability that a depth-$o(\sqrt{n})$ decision tree $T$ can distinguish $\bm u$ from~$\bm x$ is $o(1)$; in other words, the output distributions $T(\bm u)$ and $T(\bm x)$ are close to each other.

Using this Closeness Lemma, we break down the argument into stages, each stage
handled by the Stage Lemma. In each stage, we start with a subcube $X \subseteq
\zo^n$ where at least half the strings $x \in X$ satisfy $f(x) = 0$. A
subcube $X$ is equivalent to a partial assignment of variables, where a set $F
\subseteq [n]$ are ``free''\!, while variables $[n] \setminus F$ are ``fixed'',
i.e., all $x \in X$ agree on the coordinates $[n] \setminus F$. Our goal in each
stage is to find $\approx\sqrt{n}$ free variables to fix (\ie~to remove from $F$
to obtain a new subcube); essentially, these newly fixed variables are the ones
queried by the decision tree, along with the $\sqrt n$ ``sprinkled'' 1-bits. Our
choice should satisfy the property that, after fixing these new variables, we
maintain the invariant that at least half the remaining strings $x$ have $f(x) =
0$. If we can accomplish this goal in each stage, repeating this process over
$0.9 \sqrt{n}$ stages yields a contradiction as illustrated in
\cref{fig:process-illustration}: on one hand, we have fixed $0.9\sqrt n \cdot
\sqrt n = 0.9n$ bits to $1$, so $f$ is identically $1$ on all remaining strings.
On the other hand, we maintained the invariant that half the remaining strings
$x$ have $f(x) = 0$, forcing a contradiction.

We find the $\approx \sqrt n$ variables to fix as follows. The key trick is that, by Yao's principle,
there is a deterministic $o(\sqrt n)$-depth decision tree $T$ that errs with probability at most $\varepsilon$ on the distribution
\[
    \tfrac12 ( \unif(X) + \sigma ),
\]
where $\sigma$ is the ``sprinkled-1s'' distribution: start with a uniform $\bm u
\sim X$ and construct $\bm x$ by fixing a random set of $\sqrt n$ free variables
to 1. Crucially, $T$ has error probability $2\varepsilon$ over each component
distribution $\unif(X)$ and $\sigma$. Since $\Pr_{\bm x \sim X}[ f(\bm x) = 0 ] \geq 1/2$
and $T$ has error $2\epsilon$ over $\unif(X)$, we can easily find a 0-leaf
$\ell$ where $f(x) = 0$ for most $x$ reaching that leaf; by fixing the variables
queried by $T$ on the path to $\ell$, we obtain a subcube $L$ with $\Pr_{\bm x
\sim L}[ f(\bm x) = 0 ] \geq 1/2$. We can ``upgrade'' the last step to let us fix
an additional $\sqrt n$ variables to 1, by using the fact that $T$ also has error $2\varepsilon$
over $\sigma$, and the Closeness Lemma, which says that the distribution over leaves of $T$ is
similar for $\sigma$ as for $\unif(X)$.

\ignore{
We find the assignment as follows. Suppose without loss of generality that initially $\Pr[f(\bm u) = 0] \ge 1/2$. Let $\sigma$
be the distribution of $\bm x$. By Yao's principle, there exists a deterministic $o(\sqrt{n})$-depth
decision tree $T$ that errs with probability $\varepsilon$ on the distribution $\tfrac{1}{2}(\unif(\{0,1\}^n) + \sigma)$. Consequently $T$ errs with probability at most $2\varepsilon$ on $\unif(\{0,1\}^n)$ and on
$\sigma$. Using the Closeness Lemma ($T$ does not distinguish $\sigma$ and $\unif(\{0,1\}^n)$) we find a leaf
$\ell$ of $T$ such that $\Pr_{\bm x \sim \sigma}[f(\bm x) = 0 \mid \bm x \in \ell] \ge 1/2$. Observing
that $\sigma$ is a uniform mixture of $\sigma_I = \unif(\{x \in \{0,1\}^n \mid x_I = 1_I\})$ across $I
\in \binom{[n]}{\sqrt{n}}$, we see that there exists $I$ such that
\[
    \Pr_{\bm u \sim \{0,1\}^n}[f(\bm u) = 0 \mid \bm u \in \ell, \bm u_I = 1_I] =
    \Pr_{\bm x \sim \sigma_I}[f(\bm x) = 0 \mid \bm x \in \ell] \ge 1/2.
\]
Thus, fixing the values queried in leaf $\ell$ (so that $x \in \ell$) and
assigning 1s to $I$ achieves the goal.
}

\paragraph{Full proof (skippable on first reading).}
We now give a more formal proof that serves as a blueprint for our communication lower bound. 
The following lemma formalizes the process for each stage, that is, the act of assigning $(1+o(1))\sqrt{n}$
bits. To better match with the communication proof we view the conditioning on a partial assignment as
zooming in to a subcube of $\{0, 1\}^n$. 

\emph{Notation:} Throughout the paper, for any string $x \in \Sigma^n$ and a set
$S \subseteq [n]$, we write $x_F \in \Sigma^F$ for the substring of $x$ on
coordinates $F$. For $X \subseteq \Sigma^n$, we write $X_F \define \{ x_F \mid x
\in X \}$.

\begin{lemma}[Stage Lemma for Decision Trees]
    \label{lemma:dt-stage}
    Suppose $f \colon \zo^n \to \zo$ is a completion of $\GapMaj$ with randomized
    query complexity $o(\sqrt n)$.
    Let $X \subseteq \{0, 1\}^n$ be a subcube of the boolean cube, that is, there exists $F \subseteq [n]$ be
    such that $X_{[n] \setminus F}$ is fixed and $X_F$ is free. Suppose that $|F| \geq 0.1n$ and $
    \Pr_{\bm u \sim X}[f(\bm u) = 0] \geq 1/2$. Then there exists a subcube $X' \subseteq X$ and $F' \subseteq F$ such that
    \begin{equation*}
        \label{eq:mostly-0}
        \Pr_{\bm{u'} \sim X'}[f(\bm{u'}) = 0] \geq \frac{1}{2}
        \qquad\text{and}\qquad
        X_{F'}\enspace \text{is free}.
    \end{equation*}
    Moreover, there is a set $R \subseteq F \setminus F'$ such that: $X'_R = 1_R$ and $|R| \geq
    0.9|F \setminus F'|$.
\end{lemma}

This lemma suffices to prove the lower bound, as follows (see
\cref{fig:process-illustration} for an illustration):

\begin{proof}[Proof of \cref{thm:query-main} given \cref{lemma:dt-stage}]
    We may assume that $|f^{-1}(0)| \ge 2^{n-1}$, as
    otherwise we can swap the roles of 0/1 output values in the upcoming argument.
    Begin with
    $X = \{0, 1\}^n$, $F = [n]$, and $S = \varnothing$. We will apply the stage lemma
    iteratively as long as $|F| \geq 0.1n$ and update $X, F, S$ in each stage,
    as follows: assuming the invariant $\Pr_{\bm u \sim X}[ f(\bm u) = 0 ] \geq 1/2$, we apply the
    stage lemma and update
    \begin{itemize}[itemsep=0pt]
        \item $X \gets X'$ with $X'$ from the stage lemma, maintaining the invariant $\Pr_{\bm u' \sim X'}[ f(\bm u') = 0 ] \geq 1/2$;
        \item $F \gets F'$ from the stage lemma;
        \item $S \gets S \cup R$ from the stage lemma, maintaining the invariant that all $x \in X'$
        have value 1 on coordinates $S$.
    \end{itemize}
    On termination, we have preserved the invariant $\Pr_{\bm u \sim X}[f(\bm u)
    = 0] \geq 1/2$, but on the other hand, since $|S| \geq 0.9|[n] \setminus F|
    \geq 0.8 n$ and all $x \in X$ have coordinates in $S$ fixed to 1, we have
    $f(x) = 1$ for all $x \in X$, a contradiction.
\end{proof}

\begin{figure}[h]
    \centering
    \subfloat[Evolution of subcube $X$; it ends up contained in $f^{-1}(1)$.]{\input{pics/cube-decision.tex}}
    \hspace{1cm}
    \subfloat[Evolution of free variables $F$; each stage fixes at least $\sqrt n$ new 1s.]{\input{pics/assign-decision.tex}}
    \caption{Decision tree stages}
    \label{fig:process-illustration}
\end{figure}
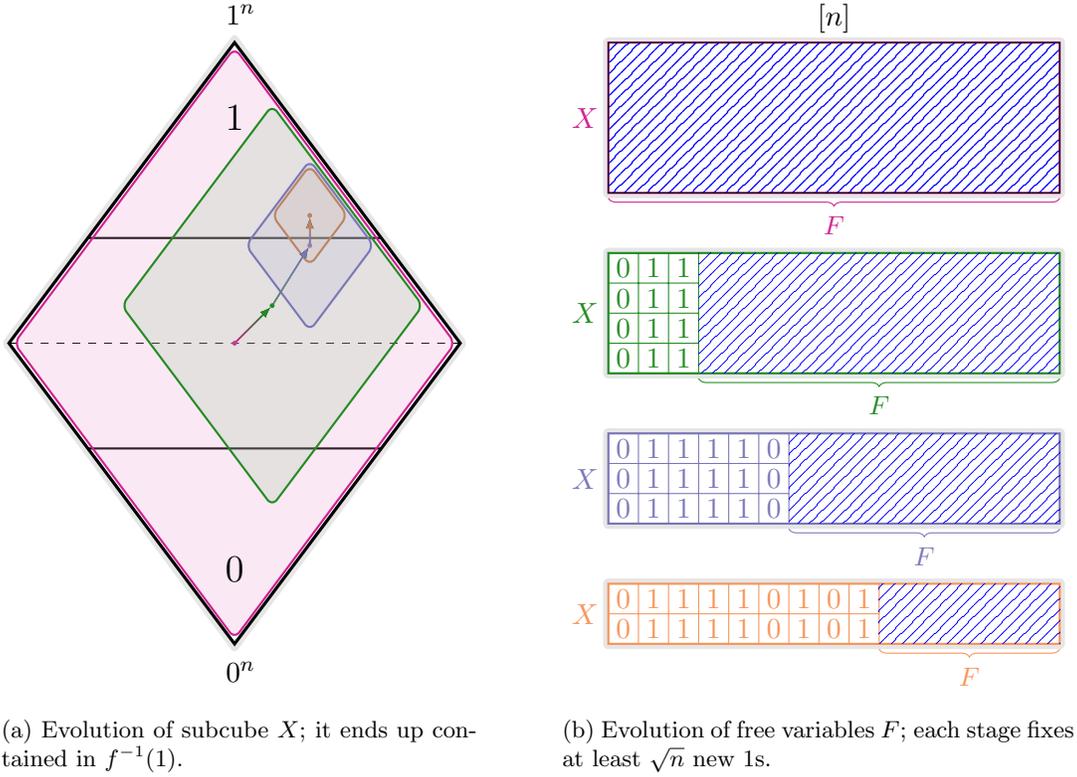

To prove the stage lemma, we define the following distribution over inputs.
Given any subcube $X \subseteq \zo^n$ with free variables $F \subseteq [n]$, we
define $\sigma$ as the distribution over $X$ obtained by choosing a uniformly
random $\bm u \sim X$ and a uniformly random set $\bm I \sim \binom{F}{\sqrt n}$
of $\sqrt n$ free variables. Then we take $\bm x \in X$ to be the string
obtained from $\bm u$ by setting all coordinates $i \in \bm I$ to 1.

This distribution has two important properties: first, it has $\sqrt n$ ``planted'' 1-valued
coordinates, so that it is more biased towards 1 than a uniformly random string. Second,
it is indistinguishable from the uniform distribution by $o(\sqrt n)$-query decision trees,
as established in the \emph{Closeness Lemma}:

\begin{lemma}[Closeness Lemma for Decision Trees]
    \label{lm:closeness-decision-trees}
    Let $X' \subseteq X$ be a subcube of $X$ and $F' \subseteq F$ the set of its free
    variables. Suppose $|F \setminus F'| \leq o(\sqrt{n})$ and $|F| \geq 0.1n$. Then
    \[
        \Pr_{\bm x \sim \sigma}[\bm x \in X'] \geq (1 - o(1))\Pr_{\bm u \sim X}[\bm u \in X'] .
    \]
\end{lemma}
\begin{proof}
Let $\bm I \sim \binom{F}{\sqrt n}$ be as in the definition of $\sigma$. Then
$\bm x \sim \sigma$ is obtained by choosing $\bm u \sim X$ and then fixing
variables in $\bm I$ to 1. If $\bm u \in X'$ and the set $\bm I \subseteq F$ does not
intersect the set $F \setminus F'$ of variables which are fixed by $X'$, then
it must be the case that $\bm x \in X'$ as well. Therefore
\[
    \Pr_{\bm x \sim \sigma}[ \bm x \in X' ]
        \geq \Pr_{\bm u \sim X}[ \bm u \in X' \wedge (\bm I \cap (F' \setminus F) = \varnothing) ]
        = \Pr_{\bm u \sim X}[ \bm u \in X' ] \cdot \Pr[ \bm I \cap (F' \setminus F) = \varnothing ] .
\]
By the union bound,
\[
    \Pr[ \bm I \cap (F' \setminus F) \neq \varnothing ]
    \leq \sqrt n \cdot \frac{|F' \setminus F|}{|F|} \leq \sqrt n \cdot \frac{o(\sqrt n)}{0.1 n} = o(1) .
    \qedhere
\]
\end{proof}

We now conclude the lower bound by proving the stage lemma:

\begin{proof}[Proof of \cref{lemma:dt-stage}.]
    Suppose that $f$ has randomized decision tree complexity $o(\sqrt{n})$; we may assume that
    the randomized decision tree has error probability $\epsilon$. Let $X$ be any
    subcube with free variables $F$ satisfying $|F| \geq 0.1 n$. Recall the distribution $\sigma$
    over $X$ where $\bm x \sim \sigma$ is sampled by choosing $\bm I \sim \binom{F}{\sqrt n}$
    and $\bm x \sim \unif(\{x \in X \mid x_{\bm I} = 1_{\bm I}\})$.
    
    By Yao's principle, there exists a deterministic decision tree $T$
    of depth $o(\sqrt{n})$, computing $f$ with error $\varepsilon$ over the mixture distribution
    \[
        \tfrac{1}{2}( \unif(X) + \sigma ) .
    \]
    Observe that tree $T$ has error at most $2\varepsilon$ over each individual distribution
    $\sigma$ and $\unif(X)$.
    
    Note that each leaf $\ell$ of $T$ corresponds to a subcube $L \subseteq X$ obtained by fixing the $o(\sqrt n)$ bits queried by $T$ on the path to $\ell$; therefore the free variables $F_L$ of $L$
    satisfy $|F \setminus F_L| = o(\sqrt n)$. By applying the closeness lemma,
    \cref{lm:closeness-decision-trees}, to each subcube $L$ corresponding to the 0-leaves of $T$,
    we obtain
    \[
        \Pr_{\bm x \sim \sigma}[ T(\bm x) = 0 ] = \sum_{0-\text{leaf } L} \Pr_{\bm x \sim \sigma}[ \bm x \in L ] \geq (1-o(1)) \Pr_{\bm u \sim X}[ T(\bm u) = 0 ]
        \geq (1-o(1)) \left( \tfrac{1}{2} - 2\epsilon \right) \geq 1/4 ,
    \]
    where the penultimate inequality uses the assumption $\Pr_{\bm u \sim X}[ f(\bm u) = 0 ] \geq 1/2$.
    Now, since $T$ errs with probability at most $2\epsilon$ on $\sigma$, we have
    \[
        \Pr_{\bm{x} \sim \sigma}[f(\bm{x}) = 1 \mid T(\bm{x}) = 0] =
        \frac{\Pr_{\bm{x} \sim \sigma}[f(\bm{x}) = 1 \land T(\bm{x}) = 0] }{ \Pr_{\bm x \sim \sigma}[T(\bm x) = 0] }
        \le 8\varepsilon.
    \]
    Hence, according to the total probability law, there exists a $0$-leaf $\ell$ in $T$ such that $\Pr_{\bm{x} \sim \sigma}[f(\bm{x}) = 1 \mid \bm{x} \in L] \le 8\varepsilon\le 1/2$. 
    Again by total probability law, there exists $I \in \binom{F}{\sqrt{n}}$ such that $\Pr_{\bm x \sim \sigma}[f(\bm x) = 1 \mid \bm x \in L, \bm I = I] \le 1/2$. We then observe that for $\bm u \sim X$ the distributions of $(\bm u \mid \bm u_I = 1)$ and $(\bm x \mid \bm I = I)$ coincide, therefore 
      $\Pr_{\bm u \sim X}[f(\bm u) = 1 \mid \bm u \in \ell, \bm u_I = 1] \le 1/2$.
    We then conclude by defining $X'\coloneqq\{x \in L \mid x_I = 1_I\}$ with $R \coloneqq I$
    being the variables forced to 1, and the new set of free variables being $F' \coloneqq F_L \setminus R$ (the free variables remaining in the leaf $L$, minus those forced to 1), which satisfies
    $|R| = \sqrt n \geq 0.9 |F \setminus F'| = 0.9 (\sqrt n + o(\sqrt n))$.
\end{proof}

\subsection{Proof Overview}
\label{sec:summary}

Let us describe how to lift the argument for decision trees in
\cref{section:lb-dt} to prove a lower bound for the communication complexity of
$\GapMaj \circ \IP_m^n$. For the sake of brevity, we will write $\Sigma \define
\zo^m$ for the alphabet of inputs to the $\IP_m$ gadget, so that both inputs to
$\GapMaj \circ \IP_m^n$ are in domain~$\Sigma^n$. For the remainder of the
paper, $m = \log |\Sigma|$ indicates the bit-size of the $\IP_m$ gadget. Below, we use $\RandCC(\cdot)$ to denote randomized communication complexity (to within error $1/3$).
\begin{theorem}
\label{thm:main-detailed}
Any completion $f \colon \Sigma^n \times \Sigma^n \to \zo$ of
$\GapMaj \circ \IP_m^n$, with $m \coloneqq 100 \log n$, has
\[
    \RandCC(f) = \Omega(\sqrt n\log n).
\]
\end{theorem}

Our proof proceeds analogously to the proof for query complexity in
\cref{section:lb-dt}. At a high level, we have the following analogies:
\begin{itemize}[itemsep=0pt,leftmargin=1em]
    \item For decision trees we tracked a subcube $X \subseteq \zo^n$, because leaves of a decision tree correspond to subcubes. Now we track a rectangle $X \times Y \subseteq \Sigma^n \times \Sigma^n$
    because leaves of a communication protocol correspond to rectangles. (Gavinsky's argument
    for parity decision trees tracks an affine subspace $X \subseteq \mathbb{F}_2^n$.)
    \item We still have a set $F$ of ``free'' variables, but these take on a different meaning,
    because the rectangle $X \times Y$ no longer
    corresponds exactly to a set of inputs obtained by fixing variables $[n] \setminus F$ and leaving those in $F$ completely unrestricted; this is because communication protocols may exchange information about variables without entirely ``fixing'' them. Our ``free variables'' $F$ instead satisfy the condition that ``only a small amount of information about them is known''\!.
    \item We still have a \emph{\nameref{lemma:stage-lemma}}, which suffices to prove the theorem in almost the same
    way: starting with a rectangle $X \times Y$ where $f(x,y) = 0$ for most $(x,y) \in X \times Y$,
    we iteratively restrict to a smaller rectangle $X' \times Y'$ by ``fixing'' some free variables,
    including $\sqrt n$ variables fixed to 1. Unlike the decision tree proof, we can no longer
    fix $\sqrt n$ variables to 1 in \emph{every} stage, but only \emph{most} stages (which we call ``safe stages'').
    \item We still have a \emph{\nameref{lemma:closeness}}, which we use to find $\sqrt n$ variables to fix to 1;
    however, the simple birthday paradox argument in \cref{lm:closeness-decision-trees} no longer
    works, because a protocol does not simply query $o(\sqrt n)$ bits, so $\sqrt n$ random coordinates
    will not simply ``miss'' the queried coordinates with high probability.
    Indeed, our new Closeness Lemma will only apply in the ``safe stages''\!.
\end{itemize}
\noindent
Let's give a little more detail on how the Stage Lemma and Closeness Lemma differ from their decision
tree analogues. We also provide a diagram of the Stage Lemma in \cref{fig:stage-lemma-diagram}.

\paragraph{Stage Lemma:}
Suppose for the sake of contradiction that there exists a
randomized communication protocol with error~$\varepsilon$ and communication cost
$o(\sqrt n m)$. We proceed as follows: 
\begin{enumerate}
    \item Start each stage with a rectangle $X \times Y$ (originally $\Sigma^n
    \times \Sigma^n$) which contains mostly 0-valued inputs to $f$, and a set $F
    \subseteq [n]$ of ``free variables''. We keep track of the ``deficit''
    \[
        \deficit(X \times Y, F) \define 2|F|m - \minentropy(\bm{x}_F) - \minentropy(\bm{y}_F),
    \]
    where $(\bm x,\bm y)\sim X\times Y$ and $\minentropy(\bm x)\coloneqq \min_x\log(1/\Pr[\bm x=x])$ denotes min-entropy. Intuitively, the deficit quantifies how many bits of information Alice and Bob know about the free variables~$F$. For intuition, note that if Alice and Bob were
    merely simulating a decision tree (solving individual gadgets one at a time) then
    the deficit would always be 0.
    \item By Yao's principle we obtain a deterministic protocol $\Pi$ with cost $|\Pi| = o(\sqrt{n} m)$
    and error $\varepsilon$ over the distribution
    \[
        \frac{1}{2} \left( \unif(X \times Y) + \sigma(X \times Y, F) \right) \,,
    \]
    where $\sigma(X \times Y, F)$ is a \emph{sprinkled-1s} distribution (\cref{def:sprinkled-1s}). In the sprinkled-1s distribution, an input $(\bm x, \bm y)$
    is chosen uniformly from $X \times Y$, and then a random set of $\sqrt n$ gadgets
    in $F$ are fixed to have value 1, \ie we force the input distribution to
    increase the number of 1-valued gadgets by at least $\sqrt n$.

    \item We use the fact that $\Pi$ has error $2\varepsilon$ on each of $\unif(X \times Y)$
    and $\sigma(X \times Y, F)$, together with the Closeness Lemma (which says that distribution
    over leaves induced by $\sigma(X \times Y, F)$ and $\unif(X \times Y)$ is similar),
    to find a smaller subrectangle $X' \times Y' \subseteq X \times Y$ and smaller set of free variables $F' \subseteq F$, where $\sqrt n$ new $\IP$ gadgets are fixed to 1. However, the Closeness Lemma
    will now apply only in ``safe stages'' where the deficit is low; in the unsafe stages, we
    skip the step of fixing $\sqrt n$ gadgets to 1. This is OK, because most stages will be safe.
\end{enumerate}

\paragraph{Closeness Lemma:}
    The Closeness Lemma should show that the distribution over leaves of the protocol
    is similar under the sprinkled-1s distribution $\sigma(X \times Y, F)$ and $\unif(X \times Y)$,
    so that we can safely sprinkle $\sqrt n$ 1-values ``without the protocol noticing''\!. However, to prove this claim, there are two caveats, which constitute the main technical novelty of the proof,
    and do not have an analogue in Gavinsky's parity decision tree proof:
    \begin{enumerate}
        \item In each stage, we must first transform $\Pi$ into a more structured protocol $\Pi'$ using
        techniques from lifting and properties of the $\IP_m$ gadget.
        \item We can prove the claim only when the deficit $\deficit(X \times Y, F)$ is \emph{small}
        (\ie~there is little information known about the free variables $F$, so that intuitively they behave similarly to the completely free variables in the decision tree argument). 
    \end{enumerate}

\noindent
In the next section, we formalize the main claims required for this argument.

\section{Proof of the Main Theorem}
\label{section:main-proof}

For the sake of contradiction, we make the following assumption throughout the proof:

\begin{assumption}
\label{assumption}
    Assume that there is a completion of $\GapMaj \circ \IP_m^n$, denoted $f \colon \Sigma^n \times \Sigma^n \to \zo$,
    such that the randomized communication cost of $f$ is $o(\sqrt n m)$ and
    $\Pr_{(\bm x, \bm y)}[ f(\bm x, \bm y) = 0 ] \geq 1/2$.
\end{assumption}

As in the decision tree proof, the main theorem will follow from a Stage Lemma. To state the communication
version of the Stage Lemma, we need a few definitions.
First, we will use the \emph{deficit} as a potential function throughout the stage procedure:
\begin{definition}[Deficit]
\label{def:deficit}
For a set $F \subseteq[n]$ and a random variable $\bm x$ over $\Sigma^n$, we write
\[
    \Dinf(\bm x_F) \define |F| m - \H_\infty(\bm x_F).
\]
For a rectangle $X \times Y \in
\Sigma^n \times \Sigma^n$ and $F \subseteq [n]$ we define the \emph{deficit}:
    \[
        \deficit(X \times Y, F) \define \Dinf(\bm x_F) + \Dinf(\bm y_F)
            = 2|F| m
            - \minentropy(\bm{x}_F)
            - \minentropy(\bm{y}_F) ,
    \]
where $(\bm x, \bm y) \sim \unif(X \times Y)$.
Observe that the deficit is non-negative.
\end{definition}
Next, we will maintain the property that, for the current rectangle $X \times Y$
in the process, the sets $X$ and $Y$ are pseudorandom in sense of being
\emph{$\gamma$-spread}: 
\begin{definition}[Spread Variables]
        A random variable $\bm{x} \in \Sigma^J$ is $\gamma$-spread if $\forall I \subseteq J$ it holds that $\H_\infty(\bm{x}_I) \geq \gamma |I| m$. A set $X \subseteq \Sigma^J$ is $\gamma$-spread
        if $\bm x \sim \unif(X)$ is $\gamma$-spread.
\end{definition}

The reason we make this definition and maintain the spreadness property
throughout the stage process is that it lets us take advantage of
near-uniformity of the \textsc{Inner Product} gadget outputs on the free
variables. Let us state this lemma now to make the motivation of this definition
clear:

\begin{lemma}[{\cite[Lemma 13]{GLMWZ16}}]
\label{lemma:ip-gadget-pseudorandom}
Suppose  $m \geq 100\log n$. Let $X \times Y \subseteq \Sigma^n \times \Sigma^n$
and $F \subseteq [n]$ be such that the random variables $\bm{x}_F$ and
$\bm{y}_F$ are $0.9$-spread, for $(\bm x, \bm y) \sim \unif(X \times Y)$.
Then
\[
\forall  z \in \zo^F\colon
\qquad
    \Pr\left[ (\IP_m(\bm x_i, \bm y_i))_{i \in F} = z \right]
    \in 2^{-|F|} \cdot (1 \pm 2^{-m/20}) .
\]
\end{lemma}

\noindent
We may now state the communication version of the Stage Lemma.
\begin{lemma}[Stage Lemma]
    \label{lemma:stage-lemma}
    Assume \cref{assumption}. Let $X \times Y \subseteq \Sigma^n \times
    \Sigma^n$ be a rectangle and $F
    \subseteq [n]$, such that: $|F| \geq 0.1n$,
    \[
        \Pr_{\bm{x}, \bm{y} \sim X \times Y}[f(\bm{x}, \bm{y}) = 0] \geq 1/2,
        \text{ \qquad and \qquad $X_F$ and $Y_F$ are $0.9$-spread}.
    \]
    Then there exist $X'\times Y'\subseteq X \times Y$ and $F' \subsetneq F$, such that:
    $|F \setminus F'| \geq \Omega(\sqrt n)$,
    \[
        \Pr_{\bm{x}, \bm{y} \sim X' \times Y'}[f(\bm{x}, \bm{y}) = 0] \geq 1/2,
        \text{\qquad and \qquad $X'_{F'}$ and $Y'_{F'}$ are $0.9$-spread.}
    \]
    Moreover, there exists a set $R \subseteq F \setminus F'$ such that either
    $|R|= \sqrt n$ (``safe stage'') or $R = \varnothing$ (``unsafe stage''),
    and
    \begin{enumerate}
        \item For every $x,y \in X' \times Y'$, and every $i \in R$, $\IP_m(x_i, y_i) = 1$ (\ie the
            output of every gadget inside $R$ is fixed to 1);
            \label{item:fixing-ones}
        \item $\deficit(X' \times Y', F') \leq \deficit(X \times Y, F) + o(\sqrt{n}m)  - \Omega(|F \setminus
            (F' \cup R)| \cdot m)$. 
    \end{enumerate}
    Here, the constants under $\Omega(\cdot)$ are universal.
\end{lemma}

\begin{figure}%
    \centering
    \subfloat[Evolution of $X \times Y$]{\input{pics/cube-comm.tex}}
    \hspace{0.1cm}
    \subfloat[Evolution of $F$]{\input{pics/assign-comm.tex}}
    \caption{Communication stages. The second stage is \emph{unsafe}, we assign too many coordinates}
    \label{fig:comm-process-illustration}
\end{figure}
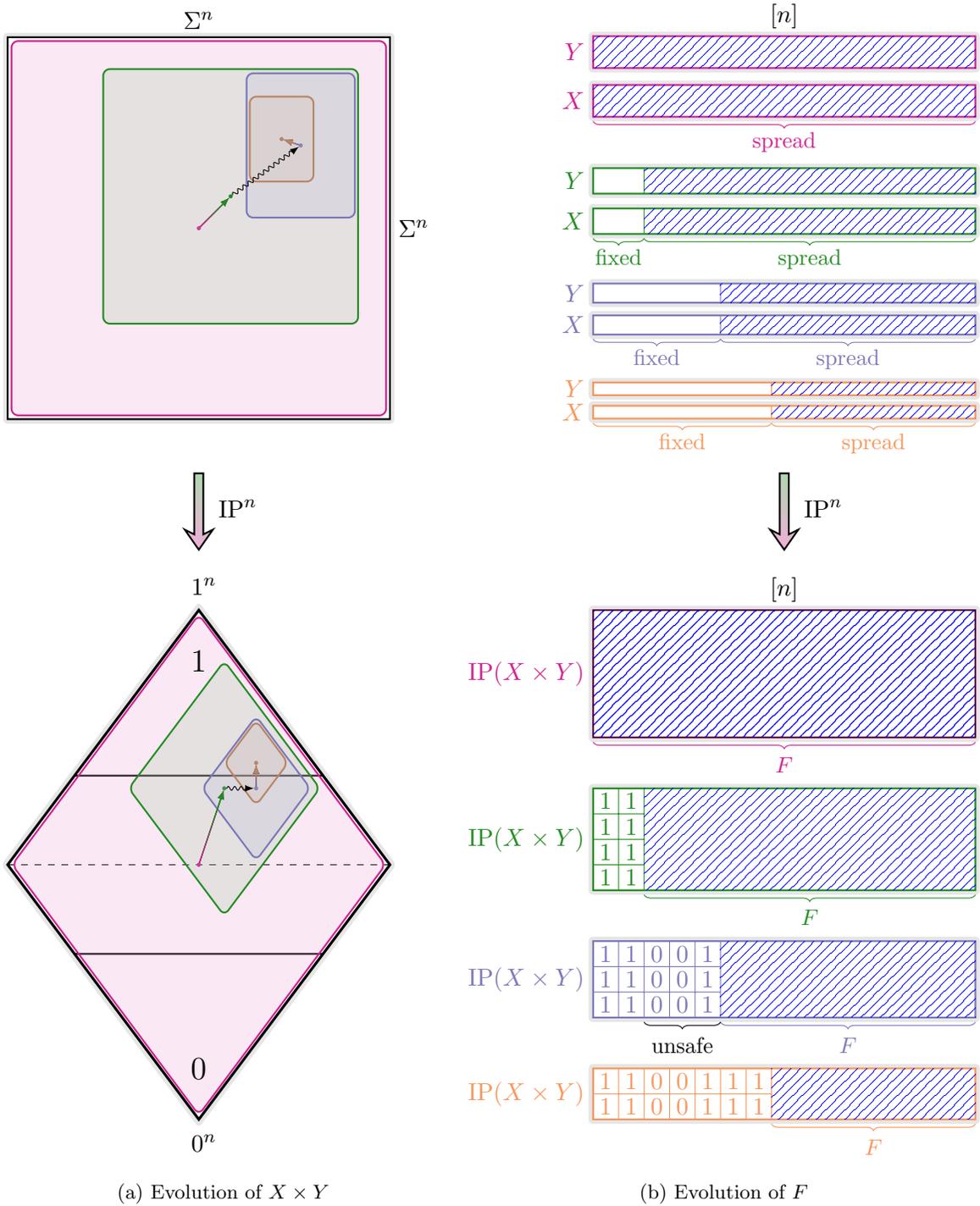

\noindent
To give some intuition behind the deficit inequality, it comes from 
\[
    \deficit(X' \times Y', F') \leq \deficit(X \times Y, F)
        + O(|\Pi| + |R|) - \Omega(|F \setminus (F' \cup R)| \cdot m) ,
\]
where $\Pi$ refers to the protocol obtained from Yao's principle, and the
deficit increases in each stage by $O(|\Pi| + |R|)$ because the entropy is
reduced by $O(1)$ for each bit of communication, as well as for each gadget in
$R$ that we fix to 1 (we lose $O(1)$ bits of entropy instead of $O(m)$ because
we leave the $2m$ input bits random); meanwhile, the process will also fix the
$2m$ input bits for each gadget in $F \setminus (F' \cup R)$ in an uncontrolled
way, to ``restore'' spreadness and decreases the deficit. 

The remainder of this paper is dedicated to proving the Stage Lemma. Assuming
this lemma, here is the proof of the main theorem:

\begin{proof}[Proof of \cref{thm:main-detailed} given \cref{lemma:stage-lemma}]
As in the proof for decision trees, 
we may assume without loss of generality that $\Pr_{\bm x, \bm y \sim \Sigma^n \times \Sigma^n}[f(\bm x,
\bm y) = 0] \geq \frac{1}{2}$, and we will
apply the Stage Lemma iteratively over many stages,

Observe that the initial value $X \times Y = \Sigma^n \times \Sigma^n$ and $F =
[n]$ satisfies the condition of \cref{lemma:stage-lemma}, and that $X' \times
Y'$ and $F'$ obtained from \cref{lemma:stage-lemma} maintain the conditions to
reapply the lemma, as long as $|F'| \geq 0.1n$. We may therefore perform the
following procedure.

Initialize
\begin{itemize}[noitemsep]
    \item $X \times Y = \Sigma^n \times \Sigma^n$ (the current rectangle);
    \item $F = [n]$ (the current ``free variables'');
    \item $T = \emptyset$ (the current set of gadgets fixed to 1);
    \item $U = \emptyset$ (the current set of ``unfree variables'' other than $T$).
\end{itemize}

While $|F| \geq 0.1n$, apply \cref{lemma:stage-lemma} to obtain $X' \times Y'
\subseteq X \times Y$, $F' \subsetneq F$, and $R \subseteq F \setminus F'$, and
update
\[
    X \times Y \gets X' \times Y',
    \qquad
    U \gets U \cup (F \setminus (F' \cup R)),
    \qquad
    F \gets F',
    \qquad
    T \gets T \cup R.
\]
Since $F$ shrinks in each iteration, the process halts. Upon halting,
the final rectangle $X \times Y$ satisfies
\begin{equation}
\label{eq:main-proof-final-rectangle-0}
    \Pr[ f(\bm x, \bm y) = 0 ] \geq 1/2
\end{equation}
for uniformly random $(\bm x, \bm y) \sim \unif(X \times Y)$.

Note that the total number of stages is $O(\sqrt{n})$, since $|F \setminus F'| = \Omega(\sqrt{n})$, and so the number of ``free variables'' decreases by $\Omega(\sqrt{n})$ on every stage.
\ignore{
\begin{claim}
    The total number of stages is $O(\sqrt{n})$.
\end{claim}
\begin{proof}
    Let us denote the total number of stages as $s$.
    Then, there are no more than $\frac{1}{2}s$ unsafe stages, which follows from the fact that a safe stage contributes at most $t = o(\sqrt{n}m)$ into the deficit, and an unsafe stage detracts at least $\Omega(\sqrt{n} m) \gg 2t$.
    It follows that at least $\frac{1}{2}s$ stages are safe, and therefore $s \leq 2
    \sqrt{n}$, since every safe stage assigns at least $\sqrt{n}$ gadgets.    
\end{proof}
}

Now consider the deficit $\deficit(X \times Y, F)$ of the final rectangle with respect to the
final set of free variables. By property (2) of \cref{lemma:stage-lemma}, this
satisfies
\[
    0 \leq \deficit(X \times Y, F) \leq o(\sqrt{n}m) \cdot O(\sqrt{n}) - \Omega(|U|m) \leq  o(n m)  - \Omega(|U| m),
\]
implying $|U| \leq o(n)$ for sufficiently
large $n$.

When the process halts, the number of free variables is at most $|F| < 0.1 n$,
meaning that $|T|+|U| = n-|F| \geq 0.9 n$, so
\[
    |T| \geq 0.9 n - |U| \geq 0.8 n .
\]
For each $i \in T$, we have ensured in property (1) of \cref{lemma:stage-lemma}
that $\IP_m(x_i, y_i) = 1$ for all $(x,y) \in X \times Y$. Since $|T| > \tfrac{2}{3} n$,
this means that $f(x,y) = 1$ for all $(x,y) \in X \times Y$, contradicting
\cref{eq:main-proof-final-rectangle-0}. This concludes the proof.
\end{proof}

\subsection{Proving the \texorpdfstring{\nameref{lemma:stage-lemma}}{}}

The structure of the proof of the Stage Lemma is shown in
\cref{fig:stage-lemma-diagram}, which explains how the 4 main ingredients fit
together: the Protocol Transformation lemma
(\cref{lemma:protocol-transformation}), the communication version of the
Closeness Lemma (\cref{lemma:closeness}), the Safe Stage Lemma
(\cref{lemma:safe-stage-lemma}), and the Unsafe Stage Lemma
(\cref{lemma:unsafe-stage-lemma}). We state these lemmas formally below. Let us
begin with the definition of the sprinkled-1s distribution, similar to the one
we used for decision trees:

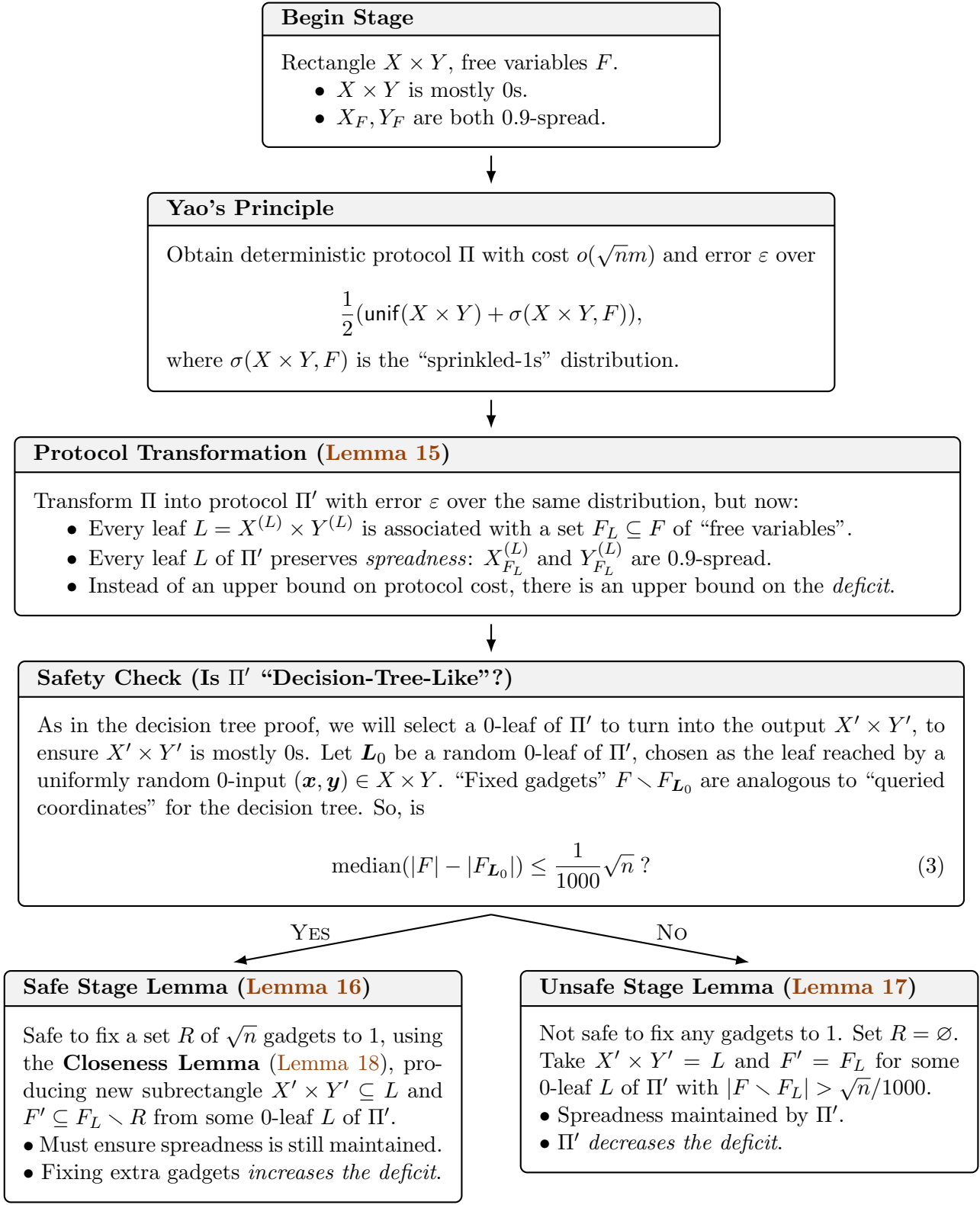
\begin{figure}
\tikzset{
  line/.style={-Latex, thick},
  block/.style={minimum width=1cm, minimum height=1.2cm, align=center, rounded corners=2pt}
}
\centering
\begin{tikzpicture}[node distance=5mm and 7mm]
  \node (start)   [block] {\begin{tcolorbox}[flowblock,width=8cm,title=Begin Stage]
      Rectangle $X \times Y$, free variables $F$.
      \begin{itemize}[noitemsep,topsep=0pt]
        \item $X \times Y$ is mostly 0s.
        \item $X_F, Y_F$ are both $0.9$-spread.
    \end{itemize}
  \end{tcolorbox}};
  \node (yao)     [block, below=of start] {%
    \begin{tcolorbox}[flowblock,width=12cm,title=Yao's Principle]
    Obtain deterministic protocol
    $\Pi$ with cost $o(\sqrt{n}m)$ and error $\epsilon$ over\\
    \[
        \frac{1}{2}(\unif(X \times Y) + \sigma(X \times Y, F)) ,
    \]
    where $\sigma(X \times Y, F)$ is the ``sprinkled-1s'' distribution.
    \end{tcolorbox}};
  \node (trans)   [block, below=of yao] {%
  \begin{tcolorbox}[flowblock,title=Protocol Transformation (\cref{lemma:protocol-transformation})]
      Transform $\Pi$ into protocol $\Pi'$ with error $\epsilon$ over the same distribution, but now:
      \begin{itemize}[noitemsep,topsep=0pt]
        \item Every leaf $L = X^{(L)} \times Y^{(L)}$ is associated with a set $F_L \subseteq F$ of ``free variables''.
        \item Every leaf $L$ of $\Pi'$ preserves \emph{spreadness}:
            $X^{(L)}_{F_L}$ and $Y^{(L)}_{F_L}$ are $0.9$-spread.
        \item Instead of an upper bound on protocol cost, there is an upper bound on the \emph{deficit}.
      \end{itemize}
  \end{tcolorbox}
  };
  \node (median)  [block, below=of trans] {%
    \begin{tcolorbox}[flowblock,title=Safety Check (Is $\Pi'$ ``Decision-Tree-Like''?)]
        As in the decision tree proof, we will select a 0-leaf of $\Pi'$ to turn
        into the output $X' \times Y'$, to ensure $X' \times Y'$ is mostly 0s.
        Let $\bm{L}_0$ be a random 0-leaf of $\Pi'$, chosen as the leaf reached by a
        uniformly random 0-input $(\bm x, \bm y) \in X \times Y$. ``Fixed gadgets''
        $F \setminus F_{\bm{L}_0}$ are analogous to ``queried coordinates'' for the
        decision tree. So, is
        \begin{equation}
        \label{eq:safe-or-unsafe}
        \mathrm{median}( |F| - |F_{\bm L_0}| ) \leq \frac{1}{1000} \sqrt n \;?
        \end{equation}
    \end{tcolorbox}};

  \node (safe)   [block] at ($ (median.south) + (-4.5cm,-3.0cm) $) {%
  \begin{tcolorbox}[flowblock,width=8cm,title=Safe Stage Lemma (\cref{lemma:safe-stage-lemma})]
  Safe to fix a set $R$ of $\sqrt n$ gadgets to 1, using the
  \textbf{Closeness Lemma} (\cref{lemma:closeness}),
  producing new subrectangle $X' \times Y' \subseteq L$ and $F' \subseteq F_L \setminus R$
  from some 0-leaf $L$ of $\Pi'$.\\
  $\bullet$ Must ensure spreadness is still maintained.\\
  $\bullet$ Fixing extra gadgets \emph{increases the deficit}.
  \end{tcolorbox}};
  \node (unsafe) [block] at ($ (median.south) + ( 4.5cm,-2.74cm) $) {%
    \begin{tcolorbox}[flowblock,width=8cm,title=Unsafe Stage Lemma (\cref{lemma:unsafe-stage-lemma})]
    Not safe to fix any gadgets to 1. Set $R = \varnothing$. Take $X' \times Y' = L$
    and $F' = F_L$ for some 0-leaf $L$ of $\Pi'$ with $|F\setminus F_L| > \sqrt n /1000$.\\
    $\bullet$ Spreadness maintained by $\Pi'$.\\
    $\bullet$ $\Pi'$ \emph{decreases the deficit}.
    \end{tcolorbox}};

  \draw[line] (start) -- (yao);
  \draw[line] (yao) -- (trans);
  \draw[line] (trans) -- (median);

  \draw[line] (median.south) -- node[pos=0.7, above]{\textsc{Yes}} (safe.north);
  \draw[line] (median.south) -- node[pos=0.7, above]{\textsc{No}}  (unsafe.north);
\end{tikzpicture}
\caption{Outline of the proof of the Stage Lemma.}
\label{fig:stage-lemma-diagram}
\end{figure}

    \begin{definition}[Sprinkled $1$s Distribution]
    \label{def:sprinkled-1s}
    For $X \times Y \subseteq \Sigma^n \times \Sigma^n$ and $F \subseteq [n]$, let
    $\sigma(X \times Y, F)$ be the \emph{sprinkled-1s distribution} obtained by
    \begin{enumerate}[noitemsep]
      \item Choosing a uniformly random subset $\bm I \sim \binom{F}{\sqrt n}$
      of free variables, and
      \item Returning $(\bm x, \bm y) \sim \unif(X \times Y)$ conditional on the event
      $\IP_m(\bm x_i, \bm y_i) = 1$ for all $i \in \bm I$.
    \end{enumerate}
    \end{definition}

\ignore{
    To find the rectangle $X' \times Y'$ on which to recurse, we start with finding a ``good leaf'' from a well-structured (deterministic) protocol $\Pi'$ computing $f$ on domain $X \times Y$. To define this protocol, start with a (deterministic) protocol $\Pi$ with cost $o(\sqrt n m)$ and error $\varepsilon$ over the distribution
    \[
        \tfrac{1}{2}\left( \unif(X \times Y) + \sigma(X \times Y, F) \right) ,
    \]
    where $\sigma(X \times Y, F)$ is the sprinkled-1s distribution in \cref{def:sprinkled-1s}. Such protocol exists by the assumption we made in the beginning of the proof. Note that the same protocol makes error no more than $2\varepsilon$ on both $\unif(X \times Y)$ and $\sigma(X \times Y, F)$.
    
    We modify the protocol $\Pi$ to obtain a protocol $\Pi'$, such that small error and small depth are preserved (at least for a typical leaf), and moreover, the leaves of $\Pi'$ are spread.
    }
    \noindent
    We prove the following Protocol Transformation lemma in  \cref{section:protocol-transformation}.
    \begin{lemma}[Protocol Transformation Lemma]
    \label{lemma:protocol-transformation}
    Let $X \times Y \subseteq \Sigma^n \times \Sigma^n$, let $F \subseteq [n]$, and
    let $\Pi$ be any deterministic protocol with cost $|\Pi| = o(\sqrt n m)$ and error $\varepsilon$ over
    an arbitrary distribution $\mu$.
    Suppose that $X_F$ and $Y_F$ are both $0.9$-spread.
    Then there exists a protocol $\Pi'$ with error at most $\varepsilon$ over the same
    distribution, such that each leaf $L = X^{(L)} \times Y^{(L)}$ of $\Pi'$ is
    associated with a subset $F_L \subseteq F$ of free variables, where:
    \begin{enumerate}[noitemsep]
        \item For every leaf $L$, the random variable $(\bm{x}^{(L)}, \bm{y}^{(L)}) \sim \unif(X^{(L)} \times Y^{(L)})$
        has $\bm{x}^{(L)}_{F_L}$ and $\bm{y}^{(L)}_{F_L}$ both $0.9$-spread; and
        \item For a random leaf $\bm{L} \define L(\bm x, \bm y)$ defined as the leaf of a random
        $(\bm x, \bm y) \sim \unif(X \times Y)$,  with probability at least $1-4\varepsilon$,
        \[
            \deficit(\bm L, F_{\bm L}) \leq \deficit(X \times Y, F) + O(|\Pi|) - \Omega(|F \setminus F_{\bm L} |  \cdot m),
        \]
        where the constants under $O(\cdot)$ and $\Omega(\cdot)$ are universal.
    \end{enumerate}
    \end{lemma}
    \ignore{
    We want to choose a leaf $L$ where $\Pi'$ outputs $0$ (to maintain the invariant that most inputs in the rectangle have value 0). After constructing $\Pi'$, we may end up
    in a bad situation where most 0-valued leaves $L_0$ reduce the number of free
    variables by a lot, say $|F|-|F_{L_0}| > \frac{1}{1000} \sqrt n$. This is how we determine whether we are in a good stage or a bad stage:
    \begin{equation}
        \mathrm{median}\left( |F| - |F_{\bm{L_0}}| \right) \leq \frac{1}{1000} \sqrt n ,
    \end{equation}
    where $F_{\bm{L_0}}$ is the set of free variables associated with the random
    0-valued leaf $\bm{L_0} \define ( \bm L \;|\; \bm L \text{ outputs 0})$. 
    }

    \begin{lemma}[Safe Stage Lemma]
    \label{lemma:safe-stage-lemma}
    Let $X \times Y \subseteq \Sigma^n \times \Sigma^n$ and $F \subseteq [n]$ satisfy the
    conditions of \cref{lemma:stage-lemma}, and assume the current stage is \emph{safe}, \ie \cref{eq:safe-or-unsafe}
    holds. Then the conclusion of \cref{lemma:stage-lemma} holds with $|R| = \sqrt n$.
    \end{lemma}

    \begin{lemma}[Unsafe Stage Lemma]
    \label{lemma:unsafe-stage-lemma}
    Let $X \times Y \subseteq \Sigma^n \times \Sigma^n$ and $F \subseteq [n]$ satisfy the
    conditions of \cref{lemma:stage-lemma}, and assume the current stage is \emph{unsafe}, \ie \cref{eq:safe-or-unsafe}
    does not hold. Then the conclusion of \cref{lemma:stage-lemma} holds with $R = \emptyset$.
    \end{lemma}

    The Stage Lemma (\cref{lemma:stage-lemma}) follows immediately from
    \cref{lemma:safe-stage-lemma,lemma:unsafe-stage-lemma}.

\subsection{Proof of the \nameref{lemma:safe-stage-lemma}}
\label{section:safe-stage-lemma}

In this section we state the main ingredients required for the
\nameref{lemma:safe-stage-lemma}, including the communication version of the
Closeness Lemma, and we prove the lemma using those ingredients. We defer the
proofs of each ingredient to \cref{section:completing-safe-stage}.

We start with the rectangle $X \times Y$ and free variables $F \subseteq [n]$.
Recall that $\Pi'$ is the transformed protocol from the
\nameref{lemma:protocol-transformation}, which has error $\epsilon$ over the
distribution
\[
    \frac{1}{2} ( \unif(X \times Y) + \sigma(X \times Y, F) ) .
\]
Suppose that the current stage is safe, \ie \cref{eq:safe-or-unsafe} holds.  We
want to find a subset $R \subseteq F$ of the free variables where we can fix all
of the $\IP$ gadget values to 1, and produce a new rectangle $X' \times Y'
\subseteq X \times Y$ and new free variables $F' \subseteq F \setminus R$
satisfying the conditions to repeat the process in \cref{lemma:stage-lemma}.

The crucial step is the ``birthday paradox'' step, where we show that we can
sprinkle $\sqrt n$ random 1-valued gadgets without the protocol $\Pi'$ noticing.
This is formalized in the next lemma, which uses spreadness of $(\bm x, \bm y) \sim \unif(X \times Y)$
on the free variables $F$, together with the pseudorandomness of the
\textsc{Inner Product} gadget under the spreadness condition
(\cref{lemma:ip-gadget-pseudorandom}).
We prove the \nameref{lemma:closeness} in \cref{sec:closeness-proof}.

\begin{restatable}[Closeness Lemma]{lemma}{lemmacloseness}
\label{lemma:closeness}
    Assume $m \geq 100\log n$.
    Let $X \times Y \subseteq \Sigma^n \times \Sigma^n$ and let $F \subseteq [n]$ be such that the random variables
    $\bm{x}_F$ and $\bm{y}_F$ are both $0.9$-spread, where $(\bm x, \bm y) \sim \unif(X \times Y)$. Now let $X' \times Y' \subseteq X \times Y$
    and $F' \subseteq F$ also be such that $\bm{x'}_{F'}, \bm{y'}_{F'}$ are both $0.9$-spread, for $(\bm x', \bm y') \sim \unif(X' \times Y')$,
    and suppose $|F| \geq 0.1n$ and $|F|-|F'| \leq \sqrt n / 1000$. Then
    \[
        \Pr_{(\bm x, \bm y) \sim \sigma(X \times Y, F)}[ (\bm x, \bm y) \in X' \times Y' ]
        \in (1 \pm 0.1) \Pr_{(\bm x, \bm y) \in \unif(X \times Y)}[ (\bm x, \bm y) \in X' \times Y' ] .
    \]
    Moreover,
    \[
    \Pr_{(\bm x, \bm y) \sim \sigma(X \times Y, F)}[ (\bm x, \bm y) \in X' \times Y' \mid \bm I\subseteq F' ]
        \in (1 \pm 0.1) \Pr_{(\bm x, \bm y) \in \unif(X \times Y)}[ (\bm x, \bm y) \in X' \times Y' ],
    \]
    where $\bm I \sim \binom{F}{\sqrt n}$ denotes the random set of coordinates used to generate $(\bm x, \bm y)$ in the distribution $\sigma(X\times Y, F)$.
\end{restatable}
\noindent
Informally, this lemma claims that if we start from a big enough spread
rectangle in $\Sigma^n \times \Sigma^n$, a typical leaf of a shallow protocol
will not be able to distinguish between uniform distribution and ``sprinkled
1s'' distribution. 

With the Closeness Lemma, finding the target rectangle $X' \times Y'$ together with
its free variables $F' \subseteq F$ and set $R \subseteq F \setminus F'$ of the $1$-valued gadgets
is done in three steps, given by the next three propositions. The first proposition claims that it is possible to find a $0$-leaf in the protocol, such that it is mostly correct under the sprinkled distribution, and the deficiency would not grow too much.
\begin{restatable}[Sprinkle 1s]{proposition}{propsprinkle}
\label{prop:safe-stage-leaf}
Let $\varepsilon > 0$ be a sufficiently small constant.
Let $X \times Y \subseteq \Sigma^n \times \Sigma^n$ and $F \subseteq [n]$
satisfy the conditions of \cref{lemma:stage-lemma}. If the current stage is
safe, then there exists a leaf $L$ of $\Pi'$ such that:
\begin{enumerate}[noitemsep]
    \item $L$ contains mostly 0 inputs under the sprinkled-1s distribution:
    \[
        \Pr_{(\bm x, \bm y) \sim \sigma(X \times Y, F)}[ f(\bm x, \bm y) = 0 \mid
            (\bm x, \bm y) \in L ] \geq 1-10\varepsilon,
    \]
    \item $|F|-|F_L| \leq  \sqrt n / 1000$.
    \item $\deficit(L, F_L) \leq \deficit(X \times Y, F) + O(|\Pi|)$.
\end{enumerate}
\end{restatable}
\noindent
The next proposition states that we can fix a particular outcome for sprinkling it $1$s; again, without increasing the deficiency too much.

\begin{restatable}[Fix 1s]{proposition}{propfixones}
\label{prop:safe-stage-clean}
Under the conditions of \cref{prop:safe-stage-leaf}, let $L$ be a leaf of $\Pi'$ obtained from
that proposition, and let $F_L$ be its associated set of free variables (as
defined in \cref{lemma:protocol-transformation}). Then there exists $X' \times
Y' \subseteq L$, $R \subseteq F_L$ with $|R|=\sqrt n$, and $F' \define F_L \setminus R$ such that
    \begin{enumerate}[noitemsep]
        \item $\Pr[ f(\bm x, \bm y) = 0 ] \geq 0.6$ for $(\bm x, \bm y) \sim \unif(X' \times Y')$; 
        \item  For all $j \in R$ and all $x,y \in X' \times Y'$,
        $\IP_m(x_j, y_j) = 1$, and
        \item $\deficit(X' \times Y', F') \leq \deficit(L, F_L) + 2|R|$.
    \end{enumerate}
\end{restatable}

\noindent
While the leaf $L$ of $\Pi'$ is spread (due to the protocol transformation), in
the the previous step we moved to a subrectangle $X' \times Y' \subseteq L$ by
fixing the $R$ gadgets to 1, which may not have preserved spreadness. We regain
spreadness by the same technique as in the \nameref{lemma:protocol-transformation}.

  \begin{restatable}[Spreadify]{proposition}{propspreadify}
  \label{prop:spreadify}
      Let $A \times B \subseteq \Sigma^n \times \Sigma^n$  and $H \subseteq [n]$ such that $\Pr_{\bm a, \bm b \sim A \times B}[f(\bm a, \bm b) = 0] \ge 0.6$. Then there exists $H' \subseteq H$ and a rectangle $A' \times B' \subseteq A \times B$ such that
      \begin{enumerate}[noitemsep]
          \item $\Pr_{\bm a, \bm b \sim A' \times B'}[f(\bm a, \bm b) = 0] \ge 1/2$.
          \item $A'_{H'}$ and $B'_{H'}$ are $0.9$-spread.
          \item $\deficit(A' \times B', H') \le \deficit(A \times B, H) - \Omega(|H \setminus H'| m)$.
      \end{enumerate}
  \end{restatable}

\noindent
We prove these propositions in \cref{section:safe-stage-propositions}. With
these tools, we can complete the proof of the safe stage lemma.

  \begin{proof}[Proof of the \nameref{lemma:safe-stage-lemma}]
  Start with $X \times Y$ and $F$ satisfying the conditions of
  \cref{lemma:stage-lemma}. Let $L$ be the leaf of $\Pi'$ with free variables
  $F_L$ obtained from \cref{prop:safe-stage-leaf}. Let $X' \times Y'$, $F'$, and
  $R = F \setminus F'$ be from \cref{prop:safe-stage-clean} and let $X'' \times
  Y'' \subseteq X' \times Y'$ and $F'' \subseteq F'$ be obtained from applying
  \cref{prop:spreadify} to $X' \times Y'$ and $F'$.

  From \cref{prop:safe-stage-clean} we guarantee that $\IP_m(x_i, y_i) = 1$ for
  every $(x,y) \in X'' \times Y''$ and every $i \in R$. From \cref{prop:spreadify,prop:safe-stage-clean} we get
  \begin{align*}
    \deficit(X'' \times Y'', F'')
        &\leq \deficit(X' \times Y', F') - \Omega(|F' \setminus F''| m) \tag{\cref{prop:spreadify}}\\
        &\leq \deficit(L, F_L) + 2|R| - \Omega(|F' \setminus F''| m)
            \tag{\cref{prop:safe-stage-clean}} \\
        &\leq \deficit(X \times Y, F) + O(|\Pi|) + 2|R| - \Omega(|F' \setminus F''| m)
            \tag{\cref{prop:safe-stage-leaf}} \\
        &\leq \deficit(X \times Y, F) + o(\sqrt{n}m) - \Omega(|F' \setminus F''| m) \tag{$|\Pi|, |R| = o(\sqrt{n}m)$} \\
        &\leq \deficit(X \times Y, F) + o(\sqrt{n}m) - \Omega(|F \setminus (F'' \cup R)|m)
  \end{align*}
  This concludes the proof of \cref{lemma:safe-stage-lemma}.
  \end{proof}

\subsection{Proof of the \nameref{lemma:unsafe-stage-lemma}}
\label{section:unsafe-stage-lemma}
A stage is \emph{unsafe} if \cref{eq:safe-or-unsafe} is false. In this case, we
choose any 0-valued leaf $L_0$ with $|F|-|F_{L_0}|$ at least the median in
\cref{eq:safe-or-unsafe}, and which also contains mostly 0-valued inputs of $f$;
we must show this exists.

Recall that $\bm{L}_0 = (\bm L \mid \bm L \text{ outputs 0})$
is a random leaf drawn by taking the unique 0-leaf containing $(\bm x, \bm y) \sim \unif(X \times Y)$ conditional on $(\bm x, \bm y)$ ending up in a 0-valued leaf.
Suppose \cref{eq:safe-or-unsafe} is false. Let $R = \emptyset$. Then, to conclude the
lemma, we want to show that
there exists a 0-valued leaf $L_0$ of $\Pi'$
with $|F|-|F_{L_0}| \geq \mathrm{median}\left(|F|-|F_{\bm{L_0}}|\right) > \sqrt{n} / 1000$, such that
\begin{itemize}
    \item $\Pr_{(\bm x, \bm y) \sim \unif(L_0)}[ f(\bm x, \bm y) = 0 ] \geq 1/2$, and
    \item $\deficit(L_0, F_{L_0}) \leq \deficit(X \times Y, F) + O(|\Pi|) - \Omega(|F \setminus F_{L_0} |  \cdot m)$.
\end{itemize}
Let $\bm{L}'_0$ be an independent copy of $\bm{L}_0$.
We upper bound the probabilities for the following events:
\begin{enumerate}[noitemsep, label=(E\arabic*)]
    \item $|F|-|F_{\bm L'_0}| < \mathrm{median}\left(|F|-|F_{\bm{L_0}}|\right)$; \label{first-event}
    \item $\textsc{error}_{\bm L'_0} > 1/2$, where $\textsc{error}_{\bm L'_0} \define \Pr_{(\bm x, \bm y) \sim \unif(\bm L'_0)}[ f(\bm x, \bm y) = 1 ]$. \label{second-event}
    \item $\deficit(\bm L'_0, F_{\bm L'_0}) > \deficit(X \times Y, F) + O( |\Pi|) - \Omega(|F \setminus F_{\bm L'_0} |  \cdot m)$. \label{third-event}
\end{enumerate}

For \ref{first-event}, by definition, $\Pr_{\bm L'_0}[|F|-|F_{\bm L'_0}| < \mathrm{median}\left(|F|-|F_{\bm{L_0}}|\right)] < 1/2$.

For \ref{second-event}, since $\Pi'$ has error $\varepsilon$ over the mixture
$\tfrac{1}{2}\left(\unif(X \times Y) + \sigma(X \times Y, F)\right)$, it has
error at most $2\varepsilon$ over $\unif(X \times Y)$. Therefore $\Exp_{\bm
L'_0}[\textsc{error}_{\bm L'_0}] \leq 4\varepsilon$,
so it follows by Markov
inequality that $\Pr_{\bm L'_0}[\textsc{error}_{\bm L'_0} \geq 20 \varepsilon]
\leq \frac{1}{5}$. For a small enough $\varepsilon$, we can conclude that:
\[
    \Pr_{\bm L'_0}\left[\textsc{error}_{\bm L'_0} > \frac{1}{2}\right] = \Pr_{\bm
L'_0}\left[\Pr_{(\bm x, \bm y) \sim \unif(\bm L'_0)}[ f(\bm x, \bm y) = 0 ] <
1/2\right] \leq \frac{1}{5} .
\]

For \ref{third-event}, we are guaranteed from \cref{lemma:protocol-transformation} that
\begin{equation}
\label{eq:unsafe-stage-deficit}
    \deficit(\bm L, F_{\bm L}) \leq \deficit(X \times Y, F) + O(|\Pi|) - \Omega(|F \setminus F_{\bm L}|m)
\end{equation}
with probability at least $1-4\varepsilon$ over a random leaf $\bm L$ chosen as the
leaf containing $(\bm x, \bm y) \sim \unif(X \times Y)$. We want this to hold
for the random 0-leaf $\bm L'_0$. Since $\Pr[ f(\bm x, \bm y) = 0 ] \geq 1/2$
and the protocol has error at most $2\varepsilon$ over $(\bm x, \bm y)$, $\bm L$ is
a 0-leaf with probability at least $1/2-2\varepsilon$. Then
\cref{eq:unsafe-stage-deficit} must hold with probability at least $1 -
10\varepsilon$ for $\bm L'_0$ instead of $\bm L$, when $\varepsilon$ is sufficiently
small.

Applying the union bound, we can conclude that with positive probability, none
of the three events occur. Then there exists a $0$-leaf $L_0$ that satisfies the
properties stated in the claim.

\section{Density Restoring Partitions and Protocol Transformation}
\label{section:protocol-transformation}

\subsection{Density Restoring Partitions}

The following is \cite[Lemma 3.5]{BPPlifting} (with straightforward adaptation
to allow for $\bm x$ non-uniform, see also \cite{liftinglowdisc}):

\begin{lemma}[Density Restoring Partition]
\label{lemma:density-restoring-partition}
    Let $\bm x \in \Sigma^n$ be a random variable with support $X$, let $F \subseteq [n]$, and let $\gamma \in (0,1)$.
    Then there exists a partition $X = \bigsqcup_{j =
    1}^r X_j$ with associated sets $I_j \subseteq F$ and values $\alpha_j \in \Sigma^{I_j}$ such that:
    \begin{itemize}
        \item $X_j \define \{ x \in X \mid x_{I_j} = \alpha_j \} \setminus \bigcup_{i < j} X_i$;
        \item $(\bm{x}_{F \setminus I_j} \mid \bm{x} \in X_j)$ is $\gamma$-spread;
        \item $\Dinf(\bm x_{F \setminus I_j} \mid \bm x \in X_j) \leq
            \Dinf(\bm x_F) - (1-\gamma)|I_j|m + \delta_j$ where
        $\smash{\delta_j = \log\left(\frac{1}{\Pr[\bm x \in \cup_{k\geq j}
        X_k]}\right)}$.
    \end{itemize}
\end{lemma}

We define a procedure to apply the density restoring partition to rectangles $X \times Y$:

\begin{lemma}[Density Restoring Partition for Rectangles]
\label{lemma:density-restoration-rectangle}
Let $X \times Y \subseteq \Sigma^n \times \Sigma^n$, $F \subseteq [n]$, $\varepsilon, \gamma
\in (0,1)$. Then there exist partitions $X = \bigsqcup_{i=1}^a X_i$ and $Y = \bigsqcup_{j=1}^b Y_j$
where each rectangle $X_i \times Y_j \subseteq X \times Y$ is associated
with a set $F_{i,j} \subseteq F$ of ``free variables''
which satisfy the following properties:
\begin{enumerate}
    \item For all $i, j$, the random variables $(\bm x_{F_{i,j}} \mid \bm x \in X_i)$ and $(\bm y_{F_{i,j}} \mid \bm y \in Y_j)$
    are both $\gamma$-spread, where $(\bm x, \bm y) \sim \unif(X \times Y)$;
    \item With probability at least $1-2\varepsilon$ over $(\bm x, \bm y) \sim \unif(X \times Y)$, the
    unique rectangle $X_{\bm{i}}\times Y_{\bm{j}}$ containing $(\bm x, \bm y)$ satisfies
    \[
        \deficit(X_{\bm i} \times Y_{\bm j}, F_{\bm{i,j}})
        \leq \deficit(X \times Y, F) - (1-\gamma)|F \setminus F_{\bm{i,j}}|m + 2\log(1/\varepsilon) .
    \]
\end{enumerate}
\end{lemma}
\begin{proof}
Begin by applying the density restoring partition of \cref{lemma:density-restoring-partition}
to each of $X$ and $Y$ to obtain
\[
    X = \bigsqcup_i X_i, \qquad\text{ and}\qquad Y = \bigsqcup_j Y_j ,
\]
where for each $i,j$ there exist sets $I_i, J_j \subseteq F$ and assignments $\alpha_i \in
\Sigma^{I_i}$, $\beta_j \in \Sigma^{J_j}$ such that
\[
    X_i \define \{ x \in X \mid x_{I_i} = \alpha_i \} \setminus \bigcup_{k < i} X_k, \qquad\text{ and }\qquad
    Y_j \define \{ y \in Y \mid y_{J_j} = \beta_j \} \setminus \bigcup_{k < j} Y_k ,
\]
where $(\bm{x}_{F \setminus I_i} \mid \bm{x} \in X_i)$ and $(\bm{y}_{F
\setminus J_j} \mid \bm{y} \in Y_j)$
are both $\gamma$-spread for $(\bm x, \bm y) \sim \unif(X \times Y)$. We
define the set $F_{i,j}$ associated with $X_i \times Y_j$ as
\[
    F_{i,j} \define F \setminus (I_i \cup J_j) ,
\]
so that the ``free variables'' are those which are not assigned by either one of the density
restoring partitions. Since $F_{i,j} \subseteq F \setminus I_i$ and $F_{i,j} \subseteq F \setminus J_j$,
the random variables
\[
    (\bm{x}_{F_{i,j}} \mid \bm{x} \in X_i) \qquad\text{ and } \qquad
    (\bm{y}_{F_{i,j}} \mid \bm{y} \in Y_j) 
\]
remain $\gamma$-spread, establishing the first property of the lemma.
Let us verify the second property. For every fixed rectangle $X_i \times Y_j$,
we have from \cref{lemma:density-restoring-partition} that
\[
    \Dinf(\bm{x}_{F \setminus I_i} \mid \bm x \in X_i)
    \leq \Dinf(\bm x_F) - (1-\gamma)|I_j|m + \log\left(\frac{1}{\Pr[ \bm x' \in \bigcup_{k \geq i} X_k ]}\right) ,
\]
where $\bm{x'} \sim \unif(X)$.
Since $F_{i,j} = F \setminus (I_i \cup J_j)$, we can use
\cref{prop:monotonicity-deficit} (stated below) to bound
$\Dinf(\bm{x}_{F_{i,j}}) \leq \Dinf(\bm{x}_{F \setminus I_i})$. Then the
remaining task to bound the final term.

\begin{claim}
With probability at least $1-\varepsilon$ over $\bm x \sim \unif(X)$, the unique
value $\bm{i} \in [a]$ such that $\bm x \in X_{\bm{i}}$ satisfies $\log\left(\frac{1}{\Pr[
\bm x' \in \bigcup_{k \geq \bm{i}} X_k ]}\right) < \log(1/\varepsilon)$.
\end{claim}
\begin{proof}[Proof of claim]
We may think of the distribution over the index $\bm i \in [a]$ defined by
$p(i) \define \Pr[ \bm i = i ] = \Pr[ \bm x \in X_i ]$. Then for all $t \in [0,1]$,
\[
    \Pr_{\bm x}\left[ \log\left(\frac{1}{\Pr[ \bm{x'} \in \bigcup_{k \geq i} X_k ]}\right) > t \right]
    = \Pr_{\bm i}\left[ \log\left(\frac{1}{\sum_{k \geq \bm i} p(k) }\right) > t \right]
    = \Pr_{\bm i}\left[ \sum_{k \geq \bm i} p(k) < 2^{-t} \right] .
\]
Let $i^*$ be the smallest value such that $\sum_{k \geq i^*} p(k) < 2^{-t}$.
If $\sum_{k \geq \bm i} p(k) < 2^{-t}$ then $\bm i \geq i^*$ so the probability of this event is at most $\sum_{k \geq i^*} p(k) < 2^{-t}$. Setting $t = \log(1/\varepsilon)$ produces the desired bound.
\end{proof}
Applying the same reasoning to $\bm y$ and using the union bound over $\bm x$ and $\bm y$,
we have with probability at least $1-2\varepsilon$ the inequality
\[
    \Dinf(\bm{x}_{F_{\bm{i,j}}} \mid \bm{x} \in X_{\bm{i}}) + \Dinf(\bm{y}_{F_{\bm{i,j}}} \mid \bm{y} \in Y_{\bm{j}})
    \leq \Dinf(\bm{x}_F) + \Dinf(\bm{y}_F) - (1-\gamma) (|I_{\bm{i}}| + |J_{\bm{j}}|)m + 2\log(1/\varepsilon),
\]
which concludes the proof since $|I_{\bm{i}}| + |J_{\bm{j}}| \geq |I_{\bm{i}} \cup J_{\bm{j}}|$.
\end{proof}

We used the following proposition in the proof above:

\begin{proposition}
\label{prop:monotonicity-deficit}
    Let $\bm x$ be a random variable over $\Sigma^n$ and let $J \subseteq F \subseteq[n]$. Then
    \[
        \Dinf(\bm{x}_{F \setminus J}) \leq \Dinf(\bm x_F) .
    \]
\end{proposition}   
\begin{proof}
First, observe
\begin{align*}
    \max_{z \in \Sigma^{F \setminus J}} \Pr[ \bm{x}_{F \setminus J} = z]
    = \max_{z \in \Sigma^{F \setminus J}} \sum_{w \in \Sigma^F, w_J = z} \Pr[ \bm{x}_F = w] 
    \leq |\Sigma|^{|J|} \max_{w \in \Sigma^F} \Pr[ \bm{x}_F = w] ,
\end{align*}
so that
\[
    \H_\infty(\bm{x}_{F \setminus J}) \geq \H_\infty(\bm{x}_F) - |J|m .
\]
Then
\[
    \Dinf(\bm{x}_{F \setminus J})
    = (|F|-|J|)m - \H_\infty(\bm{x}_{F \setminus J})
    \leq |F|m - \H_\infty(\bm{x}_F) = 
    \Dinf(\bm{x}_F) . \qedhere
\]
\end{proof}

\subsection{Protocol Transformation}

To prove the \nameref{lemma:protocol-transformation}, we must relate the growth of the deficit
function to the communication complexity of the protocol $\Pi$.
\begin{proposition}
\label{prop:protocol-deficit}
    Let $X \times Y \subseteq \Sigma^n \times \Sigma^n$, let $F \subseteq [n]$,
    and let $\Pi$ be any communication protocol with cost $d$.
    Then for any $\varepsilon > 0$,
    with probability at least $1-2\varepsilon$ over a random leaf $\bm L$ of $\Pi$
    chosen as the leaf reached by $(\bm x', \bm y')\sim X\times Y$,
    \[
        \deficit(\bm{L},F) 
        \leq \deficit(X\times Y,F) + 2 d + 2\log(1/\varepsilon).
    \]
\end{proposition}
\noindent
To prove the statement, we will need a chain rule for min-entropy:
\begin{lemma}[see e.g. {\cite[Lemma~6.30]{Vad12}}]
\label{lem:minent-chainrule}
    Let $(\bm a, \bm b)$ be distributed over $A \times B$ with $|A|=2^\ell$ and $\minentropy(\bm a, \bm b) \ge k$, and let $\mathcal{A}$ denote the marginal distribution of $\bm{a}$.
    Then for every $\varepsilon > 0$
    \[ \Pr_{\bm a' \sim \mathcal{A}}[\minentropy(\bm b \mid \bm a = \bm a') \ge k - \ell - \log 1/\varepsilon] \ge 1 - \varepsilon. \]
\end{lemma}
\begin{proof}
    Fix some $a\in A,b\in B$ then $\Pr[\bm a = a \land \bm b = b] / \Pr[\bm a = a] = \Pr[\bm b = b \mid \bm a = a]$, so $\minentropy(\bm b \mid \bm a = a) \ge \minentropy(\bm a, \bm b) - \log 1/\Pr[\bm a = a]$. Then defining $p(a) \coloneqq \Pr[\bm a = a]$ we have
    \[ \Pr[\log 1/p(\bm a) \ge \ell  + \log 1/\varepsilon] = \Pr[p(\bm a) \le \varepsilon/|A|] = \sum_{a\colon p(a) \le \varepsilon/|A|} p(a) \le \varepsilon.\qedhere\] 
\end{proof}

\begin{proof}[Proof of \cref{prop:protocol-deficit}]
The claim is equivalent to the statement that, with probability at least $1-2\varepsilon$
over a random leaf $\bm L$,
\[
    \H_\infty(\bm{x}_F \mid (\bm x, \bm y) \in \bm L)
    + \H_\infty(\bm{y}_F \mid (\bm x, \bm y) \in \bm L)
    \geq 
    \H_\infty(\bm{x}_F )
    + \H_\infty(\bm{y}_F ) - 2d - 2\log(1/\varepsilon) .
\]
Let $t(x, y)$ denote the transcript of protocol $\Pi$ on input $x,y$, so that $t(x,y) \in \zo^d$,
and observe that transcript $t(x,y)$ is in one-to-one correspondence with the leaf $L$ of $(x,y)$.
Consider the random variable $(t(\bm x, \bm y), \bm{x}_F )$.
By the chain rule for min-entropy (\cref{lem:minent-chainrule}),
with probability at least $1-\varepsilon$ over a random transcript $\bm t' = t(\bm x', \bm y')$
\[
    \H_\infty(\bm{x}_F \mid t(\bm x, \bm y) = \bm t')
    \geq \H_\infty(\bm{x}_F) - d - \log(1/\varepsilon) .
\]
Using the same argument for $\bm y$ and applying the union bound, we get the desired conclusion.
\end{proof}

We may now prove the protocol transformation lemma.

\begin{proof}[Proof of \nameref{lemma:protocol-transformation}]
Let $X \times Y \subseteq \Sigma^n \times \Sigma^n$, let $F \subseteq [n]$, and
let $\Pi$ be a deterministic protocol with cost $o(\sqrt n m)$ and error $\varepsilon$ over $\mu$.
We define the protocol $\Pi'$ in \cref{alg:transformed-protocol},
which we think of as outputting a transcript
(\ie a leaf L) along with a value. 

\begin{algorithm}[h]
\begin{algorithmic}[1]
\Input{$(x,y)\in X\times Y$}
\Output{Leaf $L'(x, y)$ and value $b=\Pi(x,y)$}
\State Run protocol $\Pi$ on $(x, y)$ to obtain a leaf $L = L(x,y)$
and a value $b$.
\State Apply the density restoring partition for rectangles (\cref{lemma:density-restoration-rectangle})
to $L = U^{(L)} \times V^{(L)}$ and $F$, to obtain
\[
    U^{(L)} = \bigsqcup_i U^{(L)}_i ,\qquad\text{ and}\qquad
    V^{(L)} = \bigsqcup_j V^{(L)}_j .
\]
\State Alice sends value $i$ for the unique $U^{(L)}_i$ containing $x$.
\State Bob sends value $j$ for the unique $V^{(L)}_j$ containing $y$.
\State The players output leaf $U^{(L)}_i \times V^{(L)}_j$ and value $b$.
\end{algorithmic}
\caption{The transformed protocol $\Pi'$}
\label{alg:transformed-protocol}
\end{algorithm}
By definition, on any input the protocol outputs the same value as the original
protocol $\Pi$, so its output is correct with the probability $1 - \varepsilon$ on $\mu$.

For each leaf $L' = U^{(L)}_i \times V^{(L)}_j$ of $\Pi'$, we associate the free
variables $F_{L'} = F^{(L)}_{i,j} \subseteq F$ given by
\cref{lemma:density-restoration-rectangle}. The first desired property of
\cref{lemma:protocol-transformation} is that for every leaf $U^{(L)}_i \times
V^{(L)}_j$ of $\Pi'$, the variables
\[
    (\bm x_{F^{(L)}_{i,j}} \mid (\bm x, \bm y) \sim U^{(L)}_i \times V^{(L)}_j )
    \qquad\text{ and }\qquad
    (\bm y_{F^{(L)}_{i,j}} \mid (\bm x, \bm y) \sim U^{(L)}_i \times V^{(L)}_j )
\]
are $0.9$-spread. This is immediately guaranteed by \cref{lemma:density-restoration-rectangle}.
The second desired property of \cref{lemma:protocol-transformation} is that
a random leaf $\bm{L'} = U^{(\bm L)}_{\bm i} \times V^{(\bm L)}_{\bm j}$ satisfies
\begin{equation*}\label{eqn:total_deficit_increase}
    \deficit(\bm L', F_{\bm L'}) 
    \leq \deficit(X \times Y, F) + O(|\Pi|) - \Omega(|F \setminus F_{\bm L'}|)m
\end{equation*}
with probability at least $1-4\varepsilon$, where constants under $O(\cdot)$ and $\Omega(\cdot)$ ($C_1$ and $C_2$, respectively) are universal. The random leaf $\bm L'$ is obtained by first selecting a random leaf $\bm
L$ of $\Pi$ by taking the leaf containing $(\bm x, \bm y) \sim \unif(X \times
Y)$, and then conditional on $(\bm x, \bm y) \in \bm L$ taking the part $U^{(\bm
L)}_{\bm i} \times V^{(\bm L)}_{\bm j}$ containing $(\bm x, \bm y)$.

By \cref{prop:protocol-deficit}, with probability at least $1-2\varepsilon$ over $\bm L$, we have
\[
    \deficit( \bm L, F )  \leq \deficit(X \times Y, F) + 2|\Pi| + 2\log(1/\varepsilon) .
\]
Now for any fixed leaf $L$ of $\Pi$, \cref{lemma:density-restoration-rectangle} guarantees
that with probability at least $1-2\varepsilon$ over $(\bm x, \bm y) \sim \unif(X \times Y)$
conditional on $(\bm x, \bm y) \in L$, we have
\[
    \deficit( U^{(L)}_{\bm i} \times V^{(L)}_{\bm j}, F^{(L)}_{\bm i, \bm j} )
    \leq \deficit( L, F ) - 0.1 |F \setminus F^{(L)}_{\bm i, \bm j}|m + 2\log(1/\varepsilon) .
\]
Therefore, by the union bound, we may combine these inequalities to obtain,
with probability at least $1-4\varepsilon$ over the random leaf $\bm L'$ of $\Pi'$,
\[
    \deficit(\bm L', F_{\bm L'})
    \leq \deficit(X \times Y, F) + 2|\Pi| - 0.1 |F \setminus F_{\bm L'}|m + 4\log(1/\varepsilon) .
\]
For any fixed $\varepsilon > 0$ and sufficiently large $n$,
$4\log(1/\varepsilon) < |\Pi|$, so we get the conclusion with $C_1 = 3$ and $C_2 = 0.1$.
\end{proof}

\section{Completing the Safe Stage Lemma}
\label{section:completing-safe-stage}

\subsection{Closeness Lemma}
\label{sec:closeness-proof}
In this section, we show that the sprinkled-$1$s distribution in
\cref{def:sprinkled-1s} is close to uniform for the ``good leaves'', \ie the
leaves $L$ where we fix at most $\sqrt n / 1000$ new variables ($|F|-|F_L| \leq
\sqrt n / 1000$). The intuition of this proof is that if the communication protocol ``fixes''
at most $\sqrt n / 1000$ variables within one stage, then this is similar to ``querying''
at most $\sqrt n /1000$ gadgets, and therefore by a birthday paradox argument
we may sprinkle a random set of $\sqrt n$ gadgets fixed to 1 without affecting the distribution over
leaves of the protocol. We require the following birthday paradox bounds:

\begin{claim}
\label{claim:birthday}
    Let $\bm I \sim \binom{[\ell]}{s}$ be a uniformly random set of $s$ elements drawn from $[\ell]$
    and let $T \subseteq [\ell]$ be such that $s \cdot |T| \leq \ell/100$. Then
    \[
        \Pr[ \bm I \cap T = \emptyset ] \geq 1 - 1/100
        \qquad\text{ and }\qquad
        \Exp[ 2^{|\bm I \cap T|} ] \leq 1 + 1/20 .
    \]
\end{claim}
\begin{proof}
    Write $\bm I = \{\bm {i}_1, \dotsc, \bm {i}_k\}$.
    For any fixed $k$ we estimate
    \[
        \Pr[ |\bm I \cap T| \geq k ]
        = \Pr[ \exists K \subseteq [s] : |K| = k \wedge \;\forall j \in K, \;\; \bm i_j \in T ] 
        \leq \binom{s}{k} \left(\frac{|T|}{\ell}\right)^k
        \leq \frac{(s \cdot |T|)^k}{\ell^k} \leq 100^{-k} .
    \]
    Then $\Pr[ \bm I \cap T = \emptyset ] \geq 1 - 1/100$ and
    \[
        \Exp[ 2^{|\bm I \cap T|} ] \leq 1 + \sum_{k = 1}^\ell 2^k \Pr[ |\bm I \cap T| \geq k ] \leq 1 + \sum_{k=1}^\ell \frac{1}{50^k}
        \leq 1 + \frac{1}{20}. \qedhere
    \]
\end{proof}

\lemmacloseness*

\begin{proof}
Within this proof we will use measure notation.
For any subset $I\subseteq F$ of size $|I|=\sqrt{n}$, let $\psi_I$ denote the measure
on rectangles defined by our distributions $\sigma(X,Y, F)$ conditioned on choosing the set $I$ of gadgets to fix to 1:
\begin{align*}
\forall A \times B \subseteq X \times Y \colon\qquad
    \psi_I(A \times B) &\define \Pr_{(\bm x, \bm y) \sim \sigma(X \times Y, F)}[ (\bm x, \bm y) \in A \times B \mid \bm I = I].
\end{align*}

According the definition of $\sigma(X,Y,F)$, the lemma's conclusion is equivalent to:
\begin{align}
\label{eq:closeness-target}
\Exp_{\bm{I}\sim \binom{F}{\sqrt{n}}}[\psi_{\bm I}(X'\times Y')] \in \frac{|X' \times Y'|}{|X \times Y|} (1 \pm 0.1 )&,\text{ and}\\
\label{eq:conditioned-closeness-target}
\Exp_{\bm{I}\sim \binom{F}{\sqrt{n}}}[\psi_{\bm I}(X'\times Y')\mid \bm{I}\subseteq F'] \in \frac{|X' \times Y'|}{|X \times Y|} (1 \pm 0.1 )&.
\end{align}

To prove the equations, we first observe that $\psi_I$ is uniform over its support
\[
    \supp(\psi_I) = \{ (x, y) \in X\times Y \mid \forall i \in I, \IP_m(x_i,y_i) = 1 \},
\]
so
\[
\psi_I(X'\times Y')=\frac{|\supp(\psi_I)\cap X'\times Y'|}{|\supp(\psi_I)|}.
\]

Next, we calculate a formula for $|\supp(\psi_I)\cap A\times B|$ for any rectangle $A\times B\subseteq X\times Y$.
For convenience, we define $\delta = 2^{-m/20} = n^{-5}$ for $m = 100\log n$.
Observe that $(1-2^{-m})^{\sqrt n} \geq 1 - \sqrt n \cdot 2^{-m} \geq 1 - \delta$.
Consider $(\bm a, \bm b) \sim \unif(A \times B)$. Let's assume that
$\bm a_S, \bm b_S$ are both $0.9$-spread, for some set $S \subseteq F$.
Then, due to
the pseudorandomness lemma, \cref{lemma:ip-gadget-pseudorandom}, the sequence of inner
products $(\IP_m(\bm a_i, \bm b_i))_{i \in S}$ is nearly uniform; specifically,
\[
    \forall z \in \zo^S \colon\qquad
    \Pr[ (\IP_m(\bm a_i, \bm b_i))_{i \in S} = z ] \in \frac{1}{2^{|S|}}(1 \pm 2^{-m/20}) = \frac{1}{2^{|S|}}(1 \pm \delta) .
\]
Therefore,
\begin{equation}
\label{eq:closeness-I-notin-S}
\begin{aligned}
    |\supp(\psi_I) \cap A \times B|
        &= |A \times B| \cdot \Pr[ (\bm a, \bm b) \in \supp(\psi_I) ] \\
        &= |A \times B| \cdot \Pr[ ( \IP_m(\bm a_i, \bm b_i) )_{i \in I} = 1^I ] \\
        &\leq |A \times B| \cdot \Pr[ ( \IP_m(\bm a_i, \bm b_i) )_{i \in I \cap S} = 1^{I \cap S} ] \\
        &\leq |A \times B| \cdot 2^{-|I \cap S|} \cdot (1 + \delta).
\end{aligned}
\end{equation}
In the case $I \subseteq S$, repeating the above calculation also yields
\begin{equation}
\label{eq:closeness-I-in-S}
    |\supp(\psi_I) \cap A \times B| \in |A \times B| \cdot 2^{-|I|} \cdot (1 \pm \delta ) .
\end{equation}

With these calculations, we may now complete the proof. For the rectangle $X \times Y$,
$(\bm x, \bm y) \sim \unif(X \times Y)$ have $\bm{x}_F, \bm{y}_F$ both $0.9$-spread
and that $I \subseteq F$ always, so
by \Cref{eq:closeness-I-in-S}, we have
\[
    |\supp(\psi_I)|=|\supp(\psi_I)\cap X\times Y|
        \in |X\times Y|\cdot 2^{-|I|}\cdot (1\pm \delta),
\]
using the fact that $\delta < 0.1$.
For the rectangle $X' \times Y'$, the random
variables $\bm{x'}_{F'}$, $\bm{y'}_{F'}$ are both $0.9$-spread for
$(\bm x', \bm y') \sim \unif(X' \times Y')$.
When $I\subseteq F'$, using the same argument, we have
\[
|\supp(\psi_I)\cap X'\times Y'|
        \in |X'\times Y'|\cdot 2^{-|I|}\cdot (1\pm \delta).
\]
Therefore, the left-hand side of~\Cref{eq:conditioned-closeness-target} can be bounded as
\[
\mathbb{E}_{\bm{I}\sim \binom{F}{\sqrt{n}}}[\psi_{\bm I}(X'\times Y')\mid \bm{I}\subseteq F']=\Exp_{\bm{I}\sim \binom{F}{\sqrt{n}}}\left[\frac{|\supp(\psi_{\bm I})\cap X'\times Y'|}{|\supp(\psi_{\bm I})|}\mid \bm{I}\subseteq F'\right]\in \frac{|X'\times Y'|}{|X\times Y|}\cdot (1\pm 2\delta),
\]
as desired. To prove~\Cref{eq:closeness-target}, we can similarly express the left-hand side as
\begin{equation}\label{eq:expt_psi}
\begin{aligned}
\Exp_{\bm{I}\sim \binom{F}{\sqrt{n}}}[\psi_I(X'\times Y')]&=\Exp_{\bm{I}\sim \binom{F}{\sqrt{n}}}\left[\frac{|\supp(\psi_{\bm I})\cap X'\times Y'|}{|\supp(\psi_{\bm I})|}\right]\\
&\in \Exp_{\bm{I} \sim \binom{F}{\sqrt{n}}}[|\supp(\psi_{\bm{I}})\cap X'\times Y'|]\cdot \frac{2^{\sqrt{n}}}{|X \times Y|}\cdot (1\pm \delta).
\end{aligned}
\end{equation}

In this case, we do not have $\bm I \subseteq F'$
always, but we do have it with large enough probability, so we can write the following lower bound
\[
    \Exp_{\bm{I} \sim \binom{F}{\sqrt{n}}}[|\supp(\psi_{\bm{I}})\cap X'\times Y'|]
    \geq \Pr[ \bm I \subseteq F' ] \cdot \frac{|X' \times Y'|}{2^{\sqrt{n}}} \cdot (1-\delta)
    \geq (1-1/100) \cdot \frac{|X' \times Y'|}{2^{\sqrt{n}}} \cdot (1 - \delta) ,
\]
where we use the fact that $|\bm I | \cdot |F \setminus F'| \leq n / 1000 \leq
|F|/100$, together with \cref{claim:birthday},
to obtain
\[
    \Pr[ \bm I \subseteq F' ] = \Pr[ \bm I \cap (F \setminus F') = \emptyset ] \geq 1 - \frac{1}{100} .
\]
Using \Cref{eq:closeness-I-notin-S}
and \cref{claim:birthday},
we also get an upper bound
\[
    \Exp_{\bm{I} \sim \binom{F}{\sqrt{n}}}[|\supp(\psi_{\bm{I}})\cap X'\times Y'|]
    \leq \frac{|X' \times Y'|}{2^{\sqrt{n}}} \cdot (1+\delta) \cdot \Exp[ 2^{|\bm I \setminus F'|} ]
    \leq \frac{|X' \times Y'|}{2^{\sqrt{n}}} \cdot (1+\delta) \cdot (1 + \tfrac{1}{20}) .
\]
By combining these bounds with \Cref{eq:expt_psi}, we
establish \Cref{eq:closeness-target}.
\end{proof}

\subsection{Sprinkling and Fixing the 1s}
\label{section:safe-stage-propositions}

Throughout this section, we are in the context of the
\nameref{lemma:safe-stage-lemma}, meaning that we have rectangle $X \times Y$
and free variables $F \subseteq [n]$ satisfying the conditions of the
\nameref{lemma:stage-lemma}; $\Pi$ is the protocol assumed in \cref{assumption},
$\Pi'$ is the transformed protocol from the
\nameref{lemma:protocol-transformation}, and  \cref{eq:safe-or-unsafe} holds.

\subsubsection{Finding a Leaf}

\propsprinkle*
\begin{proof}
We say a leaf $L$ of the protocol $\Pi'$ is \emph{good} if it satisfies the following three conditions:
\begin{enumerate}[noitemsep]
    \item $F$ is a 0-leaf of $\Pi'$ (\ie the protocol outputs 0 if it reaches $L$);
    \item $|F| - |F_L| \leq \sqrt n / 1000$;
    \item $\deficit(L, F_L) \leq \deficit(X \times Y, F) + O(|\Pi|) - \Omega(|F \setminus F_L| m)$,
    from \cref{lemma:protocol-transformation}.
\end{enumerate}
Let $L_1, \dotsc, L_N$ be all the good leaves of $\Pi'$, with associated variables $F_{L_i}$. First, we lower bound the probability to be in a good leaf, and then prove that a typical good leaf does not contain too many errors, which yields the proposition.

Let $\bm L$ be a random leaf of $\Pi'$, chosen as the unique leaf containing
$(\bm x, \bm y) \sim \unif(X \times Y)$. The current stage is safe, meaning that
\cref{eq:safe-or-unsafe} holds, so $\Pr[|F| - |F_{\bm L}| \le \sqrt{n}/1000 \mid
\text{$\bm L$ prints $0$}] \ge 1/2$. By \cref{lemma:protocol-transformation}, 
\[
    \Pr[ \deficit(\bm L, F_{\bm L}) \leq \deficit(X \times Y, F) + O(|\Pi|) - \Omega(|F \setminus F_{\bm L}| m) ]
    \geq 1-4\varepsilon.
\]
Since $\Pi'$ errs with probability at most $\varepsilon$ over the mixture distribution
$\tfrac{1}{2} \left( \unif(X \times Y) + \sigma(X \times Y, F) \right)$, it errs with probability
at most $2\varepsilon$ over
$\unif(X \times Y)$. By assumption, $\Pr_{(\bm x,\bm y) \sim \unif(X \times Y)}[f(\bm x, \bm y) = 0] \ge 1/2$, so by the union bound,
\[
    \Pr_{(\bm x, \bm y) \sim \unif(X \times Y)}[L(\bm x, \bm y) \text{ is good}]
    \geq \Pr[|F| - |F_{L(\bm x, \bm y)}| \le \sqrt{n}/1000 \text{ and } \Pi'(\bm x, \bm y) = 0] - 4\varepsilon
        \ge 1/4 - 5\varepsilon .
\]
We rewrite the inequality as
\begin{equation}
\label{eq:sprinkle-0-mass}
    \sum_{i \in [N]} \Pr_{(\bm x, \bm y) \sim X \times Y}[(\bm x, \bm y) \in L_i] \ge 1/4 - 5\varepsilon.
\end{equation}
For brevity, we write $\sigma \define \sigma(X \times Y, F)$.
By the \nameref{lemma:closeness}, for every $i \in [N]$,
\begin{align}
\label{eq:sprinkle-closeness}
    \Pr_{(\bm u, \bm v) \sim \sigma}[(\bm u, \bm v) \in L_i]
        \in (1 \pm 1/10) \Pr_{(\bm x, \bm y) \sim X \times Y}[(\bm x, \bm y) \in L_i],\\
\label{eq:conditional-sprinkle-closeness}     
\Pr_{(\bm u, \bm v) \sim \sigma}[(\bm u, \bm v) \in L_i\mid \bm{I}\subseteq F_{L_i}]
        \in (1 \pm 1/10) \Pr_{(\bm x, \bm y) \sim X \times Y}[(\bm x, \bm y) \in L_i].
\end{align}
Let $p_L \define  \Pr_{(\bm u, \bm v) \sim \sigma}[f(\bm u, \bm v) = 1 \mid (\bm u, \bm v) \in L]$.
Since $\Pi'$ has error $\varepsilon$ over the mixture, it has error
$2\varepsilon$ over $\sigma$, so
\begin{align*}
    2\varepsilon
    &\geq \Pr_{(\bm u, \bm v) \sim \sigma}[ f(\bm u, \bm v) = 1 \wedge \Pi'(\bm u, \bm v) = 0 ] \\
    &\geq \sum_{L \text{ 0-leaf of } \Pi'} p_L \Pr_{(\bm u, \bm v) \sim \sigma}[ (\bm u, \bm v) \in L ] \\
    &\geq \sum_{i \in [N]} p_{L_i} \Pr_{(\bm u, \bm v) \sim \sigma}[ (\bm u, \bm v) \in L_i ] \\
    &\geq 0.9 \sum_{i \in [N]} p_{L_i} \Pr_{(\bm x, \bm y) \sim \unif(X \times Y)}[ (\bm x, \bm y) \in L_i ] \tag{by \cref{eq:sprinkle-closeness}.}
\end{align*}
Thus, for a random $L_{\bm i}$ chosen as the unique leaf containing $(\bm x', \bm y')$
drawn from the uniform distribution on $\bigcup_i L_i$,
\begin{align*}
    \Exp[ p_{L_{\bm i}} ]&=\sum_{i\in [N]} p_{L_i}\Pr_{(\bm{x},\bm{y})\sim \unif(X\times Y)}[(\bm x,\bm y)\in L_i\mid (\bm x,\bm y)\in \bigcup_i L_i]\\
    &\leq 2\varepsilon \cdot \frac{10}{9} \cdot \Pr[ (\bm x, \bm y) \in \bigcup_i L_i ]^{-1}\\
        &\leq 2\varepsilon \cdot \frac{10}{9} \cdot \frac{1}{1/4-5\varepsilon}\tag{by \cref{eq:sprinkle-0-mass}}\\
        &\leq 10\varepsilon ,
\end{align*}
for sufficiently small $\varepsilon > 0$. Therefore, there exists a good leaf $L_i$
with $p_{L_i} \leq 10\varepsilon$. For this leaf,
\[
    \deficit(L_i, F_{L_i})
        \leq \deficit(X \times Y, F) + O(|\Pi|) - \Omega(|F|-|F_{L_i}|)m
        \leq \deficit(X \times Y, F) + O(|\Pi|),
\]
as desired.

\end{proof}

\subsubsection{Fixing the Gadgets}

\propfixones*
\begin{proof}
Let $\sigma \define \sigma(X \times Y, F)$. 
By \cref{prop:safe-stage-leaf} there exists a leaf $L = X^{(L)} \times Y^{(L)}$
of $\Pi'$ satisfying $|F| - |F_L| \le \sqrt{n}/1000$ and
\[
    \Pr_{(\bm u, \bm v) \sim \sigma}[f(\bm u, \bm v) = 1 \mid (\bm u, \bm v) \in L] \le 1/5 .
\]

Recall from \cref{def:sprinkled-1s} that $(\bm u, \bm v) \sim \sigma$ is
obtained by choosing $\bm I \sim \binom{F}{\sqrt{n}}$ and taking
$(\bm u, \bm v)$ uniform over the set $\{ (u, v) \in X \times Y \mid \IP_m(u_i, v_i) = 1, \; \forall i \in \bm I \}$.
Observe that
\begin{align*}
&\Pr_{(\bm u,\bm v)\sim \sigma}[f(\bm u,\bm v)=1\mid (\bm u,\bm v)\in L \wedge \bm I\subseteq F_L]\cdot \Pr_{(\bm u,\bm v)\sim \sigma}[\bm{I}\subseteq F_L\mid (\bm u,\bm v)\in L]\\
=&\Pr_{(\bm u,\bm v)\sim \sigma}[f(\bm u,\bm v)=1\wedge \bm I\subseteq F_L\mid (\bm u,\bm v)\in L ]\\
\le& \Pr_{(\bm u, \bm v) \sim \sigma}[f(\bm u, \bm v) = 1 \mid (\bm u, \bm v) \in L]\\
\le& 1/5.
\end{align*}
Moreover, by Bayes' rule, we have
\begin{align*}
\Pr_{(\bm u,\bm v)\sim \sigma}[\bm{I}\subseteq F_L\mid (\bm u,\bm v)\in L]&=\Pr[\bm{I}\subseteq F_L]\cdot \frac{\Pr_{(\bm u,\bm v)\sim \sigma}[(\bm u,\bm v)\in L\mid \bm{I}\subseteq F_L]}{\Pr_{(\bm u,\bm v)\sim \sigma}[(\bm u,\bm v)\in L]}\\
&\ge 0.99\cdot \frac{\Pr_{(\bm u,\bm v)\sim \sigma}[(\bm u,\bm v)\in L\mid \bm{I}\subseteq F_L]}{\Pr_{(\bm u,\bm v)\sim \sigma}[(\bm u,\bm v)\in L]}\tag{by \cref{claim:birthday}}\\
&\ge 0.99\cdot \frac{0.9\Pr_{(\bm x,\bm y)\sim X\times Y}[(\bm x,\bm y)\in L]}{1.1\Pr_{(\bm x,\bm y)\sim X\times Y}[(\bm x,\bm y)\in L]}\tag{by \cref{eq:sprinkle-closeness,eq:conditional-sprinkle-closeness}}\\
&>0.8.
\end{align*}
If follows that
\[
\Pr_{(\bm u,\bm v)\sim \sigma}[f(\bm u,\bm v)=1\mid (\bm u,\bm v)\in L \wedge \bm I\subseteq F_L]\le (1/5)/0.8=1/4.
\]
So there exists a choice $\bm I = R \in \binom{F_L}{\sqrt n}$ such that
\[
    \Pr[ f(\bm u, \bm v) = 1 \mid (\bm u, \bm v) \in L \wedge \bm I = R ] \leq 1/4 .
\]
Let $(\bm x, \bm y) \sim \unif(L)$ and let 
$(\bm u', \bm v')$ be distributed identically to $(\bm u, \bm v \mid (\bm u,
\bm v) \in L \wedge \bm I = R )$, which is equivalent to $(\bm u', \bm v')$ distributed
as $(\bm x, \bm y \mid \IP_m(x_i, y_i) = 1, \; \forall i \in R )$. In calculations below, we write $F \define F_L$ for brevity. We first want to show
\begin{equation}
\label{eq:fix-gadget-minentropy}
    \H_\infty(\bm u'_F, \bm v'_F) \geq \H_\infty(\bm x_F, \bm y_F) - |R| - 1 .
\end{equation}
From \cref{lemma:ip-gadget-pseudorandom}, 
\[
    \Pr[ \forall i \in R :\; \IP_m(\bm x_i, \bm y_i) = 1 ] \in 2^{-|R|} (1 \pm 2^{-m/20}) ,
\]
so
\[
    \max_{(u,v) \in \Sigma^F \times \Sigma^F} \Pr[ (\bm u'_F, \bm v'_F) = (u,v) ]
        \leq \max_{(x,y) \in \Sigma^F \times \Sigma^F} \Pr[ (\bm x_F, \bm y_F) = (x,y) ] \cdot 2^{|R|}(1 - 2^{-m/20})^{-1} ,
\]
which establishes \cref{eq:fix-gadget-minentropy}. The set $\{ x,y \mid \IP_m(x_i, y_i) = 1, \forall i \in R\}$ is not
a rectangle; to obtain a rectangle we will find variable settings $x^0_R, y^0_R \in \Sigma^R$ so that the
rectangle $X' \times Y' \define \{ (x,y) \in L \mid (x_R, y_R) = (x^0_R, y^0_R) \}$ satisfies the desired condition.
Let $(\bm x^0_R, \bm y^0_R)$ be distributed identically to $(\bm u'_R, \bm v'_R)$. Then
by the chain rule for min-entropy (\cref{lem:minent-chainrule}), for all $\delta \in (0,1)$,
\[
    \Pr_{(\bm x^0_R, \bm y^0_R)}\Big[ \H_\infty\big(\bm u'_{F \setminus R}, \bm v'_{F \setminus R}
        \mid (\bm u'_R, \bm v'_R) = (\bm{x}^0_R, \bm{y}^0_R) \big) \geq \H_\infty(\bm u'_F, \bm v'_F) - 2|R|m - \log(1/\delta)\Big]
        \geq 1-\delta.
\]
The last step is to ensure that the rectangle will be mostly 0-valued. Observe that
\[
    \Exp_{\bm x^0_R, \bm y^0_R}\left[ \Pr_{\bm u', \bm v'}[ f(\bm u', \bm v') = 1 \mid (\bm u'_R, \bm v'_R) = (\bm x^0_R, \bm y^0_R) ] \right]
    = \Pr[ f(\bm u', \bm v') = 1 ] \leq 1/4 ,
\]
so by Markov's inequality we have
\[
    \Pr_{\bm x^0_R, \bm y^0_R}\left[ \Pr[ f(\bm u', \bm v') = 1 \mid (\bm u'_R, \bm v'_R) = (\bm x^0_R, \bm y^0_R) ]
        > 0.4 \right] \leq \frac{1}{4} \cdot \frac{10}{4} = 5/8.
\]
By the union bound (using $\delta = 1/4$), there exist $x^0_R, y^0_R$ such that
$\Pr[ f(\bm u', \bm v') = 0 \mid (\bm u'_R, \bm v'_R) = (x^0_R, y^0_R) ] \geq 6/10$ and
\[
    \H_\infty\left(\bm x'_{F \setminus R}, \bm v'_{F \setminus R} \mid (\bm u'_R, \bm v'_R) = (x^0_R, y^0_R) \right)
    \geq \H_\infty(\bm u'_F, \bm v'_F) - 2|R|m - 2
    \geq \H_\infty(\bm x_F, \bm y_F) - 2|R|m -|R| - 3 .
\]
Therefore the rectangle $X' \times Y'$ for $X' \define \{ x \in X^{(L)} \mid x_R = x^0_R \}$
and $Y' \define \{ y \in Y^{(L)} \mid y_R = y^0_R \}$ satisfies
\begin{align*}
    \deficit(X' \times Y', F \setminus R)
    &= 2(|F|-|R|)m
        - \H_\infty(\bm u'_{F \setminus R} \mid \bm u'_R = x^0_R)
        - \H_\infty(\bm v'_{F \setminus R} \mid \bm v'_R = y^0_R) \\
    &\leq 2(|F|-|R|)m
        - \H_\infty(\bm x_F, \bm y_F) + 2|R|m + |R| + 3 \\
    &\leq 2|F|m
        - \H_\infty(\bm x_F, \bm y_F) + 2|R| \\
    &= \deficit(L, F) + 2|R| . \qedhere
\end{align*}
\ignore{
Let $\bm x', \bm y' \sim (\bm x, \bm y \mid (\bm x, \bm y) \in L, \bm{I} = J)$, then by the definition of $\sigma$ we can view the distribution of $\bm x', \bm y'$ as uniform over $\{x,y \in L \mid \IP^J_m(x_J, y_J) = 1_J\}$. Applying \cref{lemma:ip-gadget-pseudorandom} we get $\minentropy(\bm x', \bm y') \ge \minentropy(L) - |J| - 1$. We need \cref{lem:minent-chainrule}.

\nathan{TODO: I moved the chain rule, gotta edit this to reference it properly.}
\artur{The proof doesn't correspond to the statement now.}

Hence by \cref{lem:minent-chainrule} we get that with probability at least $5/6$ over $\bm x'', \bm y'' \sim \bm x', \bm y'$ we get 
\[ \minentropy(\bm x'_{F \setminus J}, \bm y'_{F \setminus J} \mid \bm x_J = \bm x''_J, \bm y_J = \bm y''_J) \ge \minentropy(L) - |J| - 1 - |J|m - \log_2 6. \]
On the other hand 
\[1/5 \ge \Pr[f(\bm x', \bm y') = 1] = \Exp_{\bm x'', \bm y''} \Pr[f(\bm x', \bm y') = 1 \mid \bm x'_J = \bm x''_J, \bm y'_J = \bm y''_J].\]
By Markov's inequality we have $\Pr[f(\bm x', \bm y') = 0 \mid \bm x'_J = \bm x''_J, \bm y'_J = \bm y''_J] \ge 1-1/5-1/6 > 0.6$ with probability at least $3/5$. Hence, there is a setting of $x^0, y^0$ such that $X' \coloneqq \{x \in X^{(L)} \mid x_J = x^0_J\}$ and $Y' \coloneqq \{y \in Y^{(L)} \mid y_J = y^0_J\}$ satisfy the required properties. 
}
\end{proof}

\subsubsection{Restoring Spreadness}

\propspreadify*
\begin{proof}
For parameter $\varepsilon > 0$ to be fixed later,
we apply the density restoring partition for the rectangle $A \times B$ (\cref{lemma:density-restoration-rectangle}) to obtain partitions $A = \bigsqcup_i A_i$ and $B = \bigsqcup_j B_j$
with each rectangle $A_i \times B_j$ associated with a set $H_{i,j} \subseteq H$
of free variables. For each $i,j$, the random variables $\bm x_{H_{i,j}}$ and
$\bm y_{H_{i,j}}$ are $0.9$-spread, where $(\bm x, \bm y) \sim \unif(A_i \times
B_j)$. With probability at least $1-2\varepsilon$ over $(\bm x, \bm y) \sim \unif(A \times B)$,
the unique rectangle $A_{\bm i} \times B_{\bm j}$ containing $(\bm x, \bm y)$
satisfies
\begin{equation}
\label{eq:spreadify-deficit}
    \deficit(A_{\bm i} \times B_{\bm j}, H_{\bm{i,j}})
    \leq \deficit(A \times B, H) - 0.1 |H \setminus H_{\bm{i,j}}|m + 2\log(1/\varepsilon) .
\end{equation}
For each $i,j$, define
\[
    p_{i,j} \define \Pr_{(\bm a', \bm b') \sim \unif(A_i \times B_j)}[ f(\bm a', \bm b') = 1 ] .
\]
Let $(\bm i, \bm j)$ the random variable chosen as the unique values such that
$A_{\bm i} \times B_{\bm j}$ contains $(\bm a, \bm b) \sim \unif(A \times B)$.
By assumption, 
\[
    \Exp[ p_{\bm{i,j}} ] = \sum_{i,j} \Pr[ (\bm x, \bm y) \in A_i \times B_j ]
        \cdot p_{i,j} = \Pr_{(\bm a, \bm b) \sim \unif(A \times B)}[ f(\bm a, \bm b) = 1 ]
        \leq 4/10 ,
\]
so by Markov's inequality, $\Pr[ p_{\bm{i,j}} > 1/2 ] \leq 8/10$. Setting $\varepsilon = 1/20$,
the probability that $p_{\bm i, \bm j} > 1/2$ or that \cref{eq:spreadify-deficit} fails for $i,j$
is at most $8/10 + 1/10 = 9/10$, so there exists rectangle $A' \times B' = A_i \times B_j$
with free variables $H' \define H_{i,j}$ satisfying the desired conditions.
\end{proof}

\appendix
\section{Appendix: Counting Arguments and Variations of FBPP}
\label{sec:easy-separation}

\subsection{Variations of FBPP}

As mentioned in the introduction, there is a variety of natural ways to define
a class \FBPP{} of ``\BPP-search problems'', \ie search problems with efficient
randomized protocols. Here are four definitions of increasing restrictiveness:
\begin{enumerate}[label={(\roman*)}]
    \item Direct translation of the definition of $\BPP$ for decision problems. A sequence of relations $R_n \subseteq \zo^n \times \zo^n \times \zo^*$
    is in $\FBPP$ if and only if $\forall n$ there exists a randomized protocol $\Pi$
    with cost $\poly\log n$ such that $\forall x,y \in \zo^n$, $\Pr[ (x,y,\Pi(x,y)) \in R_n ] \geq 2/3$.

    The issue with this definition is that, unlike for $\BPP$, the constant
    $2/3$ is no longer arbitrary: choosing a different constant will change the
    class, because for search problems, the error cannot generally be boosted
    (see examples below). \label{item-one}
    \item Fix the issue by demanding that there is an efficient protocol achieving any error $\epsilon > 0$. The sequence $R_n \subseteq \zo^n \times \zo^n \times \zo^*$
    is in $\FBPP$ if and only if $\forall n, \forall \epsilon > 0$ there exists a
    randomized protocol with cost $\poly(\log n, 1/\epsilon)$ such that $\forall x,y, \Pr[ (x, y, \Pi(x,y)) \in R_n ] \geq 1-\epsilon$. This is the approach taken in \cite{Aaronson10,Aaronson24} for Turing machines. \label{item-two}
    \item Demand that the cost of the protocol is $\poly(\log n, \log(1/\epsilon))$ instead of $\poly(\log n, 1/\epsilon)$, so that dependence on $\epsilon$ is the same as it is for decision problems. \label{item-three}
    \item Require that solutions are not only efficient to find, but also efficient to verify. This is the approach taken by \cite{Goldreich11},
    and this is also analogous to the definition of efficient deterministic search
    problems\footnote{We remark that, confusingly, the name $\mathsf{FP}$ is used for two different classes of problems: relations where a solution can be found efficiently, and relations where a solution can be found \emph{and} verified efficiently.} in the context of $\mathsf{TFNP}$. According to this definition, 
    $R_n$ is in $\FBPP$ if and only if both the following conditions hold:
    \begin{itemize}
        \item The condition from definition \ref{item-one} above, where the relation can be computed with cost $\poly \log n$ and success probability is $2/3$; and,
        \item For every $n$ there exists a randomized \emph{verifier} protocol $V$
        which outputs a value in $\zo$ such that, for all $(x,y,z) \in \zo^n \times \zo^n \times \zo^*$, when
        Alice is given $(x,z)$ and Bob is given $(y,z)$, they can run protocol $V$
        which satisfies $\Pr[ V(x,y,z) = 1_{(x,y,z) \in R_n} ] \geq 2/3$.
    \end{itemize}
    \label{item-four}
\end{enumerate}
Note that each of these definitions is more strict than the previous one: If a
relation satisfies definition \ref{item-two} then it also satisfies definition \ref{item-one}, etc.

\paragraph*{Pseudodeterminism.}
It is most natural to compare pseudodeterminism with definition \ref{item-three}. This is
because, if a problem admits an efficient $\poly\log n$ pseudodeterministic
protocol, it also satisfies definition \ref{item-three}. With a pseudodeterministic protocol,
we can take a majority vote on the output of $O(\log(1/\epsilon))$ independent
runs; with probability $1-\epsilon$, the canonical output for given input
$(x,y)$ will be the majority of outputs.

\paragraph*{Partial boolean functions.}
Each of the above definitions is \emph{more strict} than the previous one, but
this is only true for problems with large outputs. When we consider partial
boolean functions, \ie problems where the valid outputs are single bits (the
main topic of this paper), these 4 definitions collapse into only 2, because the
first 3 definitions are equivalent. Given a protocol with the guarantees in
definition~\ref{item-one}, we can perform majority-vote error boosting. If there is
only 1 valid output for given $(x,y)$ then the majority vote will take this
value; if there are 2 valid outputs for given $(x,y)$ then the output is
valid with probability 1.

If we then take definition \ref{item-three} as the definition of $\FBPP$ for partial
functions,  our \cref{thm:main} proves $\FBPP \subsetneq \FPs$
for partial functions.

Furthermore, when the number of valid outputs is small, definition \ref{item-four} is a
special case of efficient pseudodeterminism. This is because, for relations $R_n
\subseteq \zo^n \times \zo^n \times \zo^*$ with only a small $\leq \poly\log(n)$
number of possible outputs, the verifier can be used to create a
pseudodeterministic protocol, by iterating over every possible output in a fixed
order:

\begin{claim}
Let $R_n \subseteq \zo^n \times \zo^n \times \zo^*$ be a relation that satisfies
definition \ref{item-four} above, and also $|\{ z \in \zo^* \mid \exists x,y : (x,y,z) \in R_n \}|
\leq \poly\log n$. Then there is a pseudodeterministic protocol for $R_n$ with
cost $\poly\log n$.
\end{claim}

\subsection{Separations via counting}
\label{section:counting}

Let us explain how to prove \cref{eq:separation} for definitions \ref{item-one}--\ref{item-three} of
$\FBPP$ given above, using a counting argument that works equally well for boolean circuits
and other models of computation. An initial sketch of the argument for definition~\ref{item-one}
is as follows: 
\begin{quote}
Define a relation $S\subseteq\{0,1\}^n\times\{0,1\}^n\times[3]$ by choosing $S(x,y)\coloneqq\{i:(x,y,i)\in S\}$ as a random subset of~$[3]$ of size $2$ (i.e., we forbid one random output element). On any input, outputting a random number in~$[3]$ will solve $S$ with probability~$2/3$. A simple counting argument shows that the pseudodeterministic complexity is $\Omega(n)$ with high probability over the choice of $S$.
\end{quote}
For definition \ref{item-one}, we can make this argument explicit (see \cref{sec:explicit}),
\ie we exhibit a specific relation satisfying definition \ref{item-one} but which does not
have an efficient pseudodeterministic protocol. But let us now give the full
counting argument for definition \ref{item-three}.

\begin{proposition}
    There exists a sequence $R_n \subseteq \zo^n \times \zo^n \times \zo^*$ of relations such
    that for every $\epsilon > 0$, $\RandCC_{\epsilon}(R_n) = O(\log(1/\epsilon))$,
    but which has pseudodeterministic cost $\Omega(n)$.
\end{proposition}
\begin{proof}
Fix any $n$. Let $\mathcal R$ be the set of all relations $R$ such that $|R(x,y)\cap\{0,1\}^k|=2^k-1$ for all $(x,y) \in \zo^n \times \zo^n$ and $k\in[n]$. That is, in any relation $R\in\mathcal R$ we have removed exactly one of the $2^k$ valid outputs for every output length $k\in[n]$. Each relation $R \in \mathcal R$ has randomized cost $O(\log(1/\epsilon))$ since for every $\epsilon$ we can take $k =
\lceil\log(1/\epsilon)\rceil$ and output a random value $\bm z \sim \zo^k$; then
$(x,y,\bm z) \in R$ with probability at least $1 - 2^{-k} \geq 1-\epsilon$.

Fix any function $f \colon \zo^n \times \zo^n \to \zo^*$ such that the length of
any output is at most $|f(x,y)| \leq c$ for some $c$. Let $\mathcal R_f
\subseteq \mathcal R$ be the set of relations consistent with $f$ (\ie $f(x,y)
\in R$ for all $x,y$). Then
\[
    \frac{|\mathcal R_f|}{|\mathcal R|} = \prod_{x,y} \left(1 - 2^{-|f(x,y)|}\right)
    \leq \left(1- 2^{-c}\right)^{2^{2n}} \leq e^{-2^{2n-c}} .
\]
Every pseudodeterministic protocol with cost $c$ computes one of these functions $f$.
Using Newman's theorem to bound the number of random bits in any randomized protocol
by $\log n + O(1)$, the number of pseudodeterministic protocols with cost $c$ is at most
\[
    \left(\left(2^{2^n}\right)^{2^c} \cdot (2^{c+1})^{2^c}\right)^{O(n)} .
\]
Then the fraction of relations consistent with any pseudodeterministic protocol is
at most
\[
    2^{O(n (2^{n+c} + c2^c)) - \Omega(2^{2n-c})} < 1
\]
when $c = o(n)$.
\end{proof}

\subsection{Explicit Separation for Definition \ref{item-one}}
\label{sec:explicit}

In this section, we provide an explicit relation satisfying definition \ref{item-one} but does not admit an efficient pseudodeterministic protocol.
In fact, we separate randomized communication from the stronger quantum pseudoterministic model.

\begin{theorem}
There exists an explicit two-party search problem $f\subseteq
\{0,1\}^n\times \{0,1\}^n\times \{0,1\}^2$ with randomized
complexity $O(1)$ but quantum pseudodeterministic complexity $\Omega(n)$.
\end{theorem}
Before presenting the proof, we remark that in concurrent work, Aaronson, Gur, and Li~\cite{AGL26} showed an analogous separation for query complexity. Indeed, the two-party search problem we exhibit here is a composition of their search problem with the two-bit AND gadget.

    Let $\RandCC_\epsilon$ denote the $\epsilon$-error classical communication complexity,  $\QuantCC_\epsilon$ the $\epsilon$-error quantum communication complexity, and $\QuantPsCC_\epsilon$ denote the $\epsilon$-error quantum pseudodeterministic communication complexity.
    We will omit $\epsilon$ for the above-mentioned measures when $\epsilon=1/3$.

    We first claim that it suffices to find an $f$ and two different constant error parameters $0<\epsilon_2<\epsilon_1<\frac{1}{3}$, such that $\RandCC_{\epsilon_1}(f)=O(1)$ while $\QuantCC_{\epsilon_2}(f)=\Omega(n)$.
    Indeed, if we can find such an $f$, we have $\RandCC(f)\le \RandCC_{\epsilon_1}(f)=O(1)$, and on the other hand,
    \[
    \QuantPsCC(f)\ge \Omega(\QuantPsCC_{\epsilon_2}(f)/\log (1/\epsilon_2))\ge \Omega(\QuantCC_{\epsilon_2}(f)/\log (1/\epsilon_2))=\Omega(n),
    \]
    where we used the fact that quantum pseudodeterministic complexity admits error reduction.
    
    We now construct $f$ as follows:
    Let $n=2m$ without loss of generality.
    Given $x,y\in \{0,1\}^n$ as the input for Alice and Bob respectively, we interpret $x=(x^1,x^2)$ as the concatenation of two $m$-bit strings $x^1,x^2\in \{0,1\}^m$, and similarly for $y=(y^1,y^2)$.
    Then we define $f(x,y)\coloneqq \{0,1\}^2\setminus \{(\IP(x^1,y^1),\IP(x^2,y^2))\}$ as the set of all pairs of bits excluding $(\IP(x^1,y^1),\IP(x^2,y^2))$.
    We will chose $\epsilon_1=1/4,\epsilon_2=0.1$ and prove that $\RandCC_{1/4}(f)=O(1)$, while $\QuantCC_{0.1}(f)=\Omega(n)$.

    We first show $\RandCC_{1/4}(f)=O(1)$: Observe that $|f(x,y)|=3$ for all $x,y\in \{0,1\}^n$,
    the simple protocol that outputs a uniform random pair $(\bm{u},\bm{v})\sim \{0,1\}^2$ succeeds with probability $3/4$.
    
    It remains to show $\QuantCC_{0.1}(f)=\Omega(n)$.
    For the sake of contradiction, suppose that $\QuantCC_{0.1}(f)=o(n)$.
    Then there exists a quantum protocol $\Pi$ of cost $d=o(n)$ that computes $f$ with error at most $0.1$ with respect to $\unif(\{0,1\}^n\times \{0,1\}^n)$.
    We will then construct a quantum protocol $\Pi'$ of cost $d'=O(d)=o(n)$ that computes $n$-bit $\IP$ with error at most $0.4$ with respect to $\unif(\{0,1\}^m\times \{0,1\}^m)$, a contradiction to the following folklore lower bound for the quantum distributional communication complexity of IP.
    \begin{lemma}[\cite{Kremer1995}]
    The quantum distributional communication complexity of the $n$-bit $\IP$ with respect to $\mu\coloneqq \unif(\{0,1\}^{n/2}\times \{0,1\}^{n/2})$ is $\QuantCC_{0.4}(\IP_{n/2},\mu)=\Omega(n)$.
    \end{lemma}

    Before presenting the protocol $\Pi'$, we define the following notation:
    \begin{align*}
        p_{00}\coloneqq \Pr[\bm{u}\ne \IP(\bm{x}^1,\bm{y}^1)\land \bm{v}\ne \IP(\bm{x}^2,\bm{y}^2)],\\
        p_{01}\coloneqq \Pr[\bm{u}\ne \IP(\bm{x}^1,\bm{y}^1)\land \bm{v}= \IP(\bm{x}^2,\bm{y}^2)],\\
        p_{10}\coloneqq \Pr[\bm{u}= \IP(\bm{x}^1,\bm{y}^1)\land \bm{v}\ne \IP(\bm{x}^2,\bm{y}^2)],
    \end{align*}
    where $\bm{x}=(\bm{x}^1,\bm{x}^2),\bm{y}=(\bm{y}^1,\bm{y}^2)$ are uniform $n$-bit strings and $(\bm{u},\bm{v})\coloneqq \Pi(\bm{x},\bm{y})$.
    It follows that $p_{00}+p_{01}+p_{10}=\Pr[(\bm{u},\bm{v})\in f(\bm{x},\bm{y})]\ge 0.9$.

    Next, we specify the protocol $\Pi'$ as follows:
    Alice holds $\bm{x}'\sim \{0,1\}^m$ and Bob holds $\bm{y}'\sim \{0,1\}^m$.
    To compute $\IP(\bm{x}',\bm{y}')$, they first sample two uniform random $m$-bit strings $\bm{x}'',\bm{y}''\sim \{0,1\}^m$.
    Consider the following cases:
    \begin{itemize}
        \item $p_{00}\ge 0.3$: In this case, we have either $p_{00}+p_{01}\ge 0.6$ or $p_{00}+p_{10}\ge 0.6$.
        Suppose the former holds without loss of generality, as the other case can be dealt with a similar argument.
        Both parties simulate $\Pi$ on $\bm{x}\coloneqq (\bm{x}',\bm{x}'')$ and $\bm{y}\coloneqq(\bm{y}',\bm{y}'')$, obtain $(\bm{u},\bm{v})\coloneqq \Pi(\bm{x},\bm{y})$.
        Finally, they output $\neg \bm{u}$.
        We conclude that $\Pi'$ succeeds with probability
        \[
\Pr[\Pi'(\bm{x}',\bm{y}')=\IP(\bm{x}',\bm{y}')]=\Pr[\bm{u}\ne \IP(x',y')]=p_{00}+p_{01}\ge 0.6.
        \]
        \item $p_{00}<0.3$: In this case, we have $p_{01}+p_{10}\ge 0.6$.
        Similar to the previous case, both parties simulate $\Pi$ on $\bm{x}\coloneqq (\bm{x}',\bm{x}'')$ and $\bm{y}\coloneqq(\bm{y}',\bm{y}'')$, obtain $(\bm{u},\bm{v})\coloneqq \Pi(\bm{x},\bm{y})$.
        Since $(\bm{x}'',\bm{y}'')$ is the common information between both parties, they can compute $\bm{v}'\coloneqq \IP(\bm{x}'',\bm{y}'')$ without extra communication.
        Finally, they output $\neg(\bm{u}\oplus\bm{v}\oplus\bm{v}')$.
        We conclude that $\Pi'$ succeeds with probability
        \[
        \Pr[\Pi(\bm{x}',\bm{y}')=\IP(\bm{x}',\bm{y}')]=\Pr[\bm{u}\oplus \bm{v}\oplus \IP(\bm{x}',\bm{y}')\oplus \IP(\bm{x}'',\bm{y}'')=1]=p_{01}+p_{10}\ge 0.6.
        \]
    \end{itemize}

\bigskip

\subsection*{Acknowledgements}
We thank Tom Watson and an anonymous STOC reviewer for pointing out a mistake in an earlier version of this paper.
This work was supported by the Swiss State Secretariat for Education, Research, and Innovation (SERI) under contract number
MB22.00026.

\bigskip

\DeclareUrlCommand{\Doi}{\urlstyle{sf}}
\renewcommand{\path}[1]{\small\Doi{#1}}
\renewcommand{\url}[1]{\href{#1}{\small\Doi{#1}}}
\bibliographystyle{alphaurl}
\bibliography{references}

 \end{document}

%% file: defs.tex
\usepackage{fullpage}

\usepackage[usenames,dvipsnames,svgnames,table]{xcolor}
\usepackage[colorlinks=true,linkcolor=RawSienna,citecolor=ForestGreen]{hyperref}

\newcommand*\ie{i.e.,\xspace} 

\usepackage[utf8]{inputenc}
\usepackage{amsthm}
\usepackage{thmtools}
\usepackage[normalem]{ulem}
\usepackage{mathtools,bm}
\usepackage[inline]{enumitem}
\usepackage{amsmath,amsfonts,amssymb,amsthm,hyperref,algorithm,eqparbox,array,ragged2e,stmaryrd,bbm}
\usepackage{physics}
\usepackage{xspace}
\usepackage{tikz}
\usepackage{ifthen}
\usepackage{pgfmath-xfp}
\usetikzlibrary{positioning,decorations.pathreplacing,calligraphy, arrows.meta}
\pgfmathsetseed{3434}
\usepackage{algorithm}
\usepackage{algpseudocode}

\usepackage{tocloft}
\usepackage[nottoc]{tocbibind} 

\usepackage{enumitem}
\usepackage[capitalise,nameinlink,noabbrev]{cleveref}

\usepackage{subfig}
\newsavebox{\tempbox}

\usepackage{tcolorbox}

\tcbset{
  flowblock/.style={
    boxrule=0.8pt,
    colback=white,
    colframe=black,
    fonttitle=\bfseries,
    title filled,
    colbacktitle=black!5,
    coltitle=black,
    left=6pt,right=6pt,top=6pt,bottom=6pt,
    before skip=0pt, after skip=0pt
  }
}

\input{tikz.tex}

\algnewcommand\algorithmicinput{\textbf{Input: }}
\algnewcommand\Input{\item[\algorithmicinput]}
\algnewcommand\algorithmicoutput{\textbf{Output: }}
\algnewcommand\Output{\item[\algorithmicoutput]}
\algnewcommand{\OneLineIf}[2]{
  \State \algorithmicif\ #1\ \algorithmicthen\ #2}

\newcommand{\ignore}[1]{}

\declaretheorem[name=Theorem]{theorem}
\declaretheorem[sibling=theorem,name=Lemma]{lemma}
\declaretheorem[sibling=theorem,name=Proposition]{proposition}
\declaretheorem[sibling=theorem,name=Claim]{claim}

\declaretheorem[sibling=theorem,name=Conjecture]{conjecture}

\declaretheorem[sibling=theorem,name=Assumption]{assumption}

\theoremstyle{definition}
\declaretheorem[sibling=theorem,name=Definition]{definition}

\theoremstyle{remark}

\makeatletter
\def\th@example{%
  \thm@notefont{}
  \normalfont 
}
\def\th@definition{%
  \thm@notefont{}
  \normalfont 
}
\makeatother

\theoremstyle{example}

\renewcommand{\setminus}{\smallsetminus}
\renewcommand{\hat}{\widehat}

\renewcommand{\H}{\operatorname{H}}
\newcommand{\Dinf}{\operatorname{D}_\infty}

\DeclareMathOperator{\Exp}{\mathbb{E}}

\renewcommand{\epsilon}{\varepsilon}

\newcommand{\IP}{\mathrm{IP}}

\newcommand{\poly}{\mathrm{poly}}

\newcommand{\newmeasure}[2]{\newcommand{#1}{{\textup{\sffamily #2}}\xspace}}
\newmeasure{\C}{C}

\newcommand{\GapMaj}{\textsc{GapMaj}}
\newcommand{\GapHD}{\textsc{GapHD}}
\newcommand{\RandCC}{\mathsf{R}}
\newcommand{\QuantCC}{\mathsf{Q}}

\newcommand{\QuantPsCC}{\mathsf{QPs}}

\newcommand{\supp}{\mathrm{supp}}

\newcommand{\im}{\mathrm{Im}}

\ifthenelse{0>1}{
\newcommand{\artur}[1]{\textcolor{orange}{[Artur: #1]}}
\newcommand{\mika}[1]{\textcolor{blue}{[Mika: #1]}}
\newcommand{\weiqiang}[1]{\textcolor{red}{[Weiqiang: #1]}}
\newcommand{\dima}[1]{\textcolor{green!33!black}{[Dima: #1]}}
\newcommand{\anastasia}[1]{\textcolor{magenta}{[Anastasia: #1]}}
}{
\newcommand{\artur}[1]{}
\newcommand{\mika}[1]{}
\newcommand{\weiqiang}[1]{}
\newcommand{\dima}[1]{}
\newcommand{\ziyi}[1]{}
\newcommand{\gilbert}[1]{}
\newcommand{\yaroslav}[1]{}
\newcommand{\anastasia}[1]{}
}

\let\OLDthebibliography\thebibliography
\renewcommand\thebibliography[1]{
  \OLDthebibliography{#1}
  \setlength{\parskip}{0pt}
  \setlength{\itemsep}{1.3ex}
}

%% file: tikz.tex
\usepackage{tikz}

\usetikzlibrary{
    babel,
    arrows,
	calc,
	patterns,
	intersections,
    backgrounds,
    trees,
    mindmap,
    fadings,
	decorations.pathreplacing,
	decorations.pathmorphing,
	decorations.text,
	decorations.markings,
    scopes,
	shapes,
    snakes,
    shapes,
	positioning,
    arrows.meta,
    patterns,
    patterns.meta,
    angles,
    quotes,
    math
}

\tikzfading[name = middle,
    top color = transparent!100,
    bottom color = transparent!100,
    middle color = transparent!30
]

\tikzset{
    >=latex,
    pics/rhom/.style n args = {4}{
        code = {
            \draw[
                thick, #3, fill = #3!50, fill opacity = 0.2, rounded corners = 3pt]
                (-#1, 0) -- (0, -#2) -- (#1, 0) -- (0, #2) -- cycle;
            \node[draw, fill, circle, inner sep = 0pt, minimum size = 0.05cm, #3] (#4) at (0, 0) {};
        }
    },
    linecol/.style n args = {3}{
        postaction = {
            decorate,
            decoration = {
                markings,
                mark = between positions 0 and 0.9 step 0.2pt with {
                    \pgfmathsetmacro\myval{
                        multiply(
                            divide(
                                \pgfkeysvalueof{/pgf/decoration/mark info/distance from start},
                                \pgfdecoratedpathlength
                            ),
                            120
                        )};
                    \pgfsetfillcolor{#3!\myval!#2};
                    \pgfpathcircle{\pgfpointorigin}{#1};
                    \pgfusepath{fill};
                },
                mark = at position 1 with {\arrow[#3]{latex}},
            }
        }
    },
    pics/universe-rect/.style n args = {2}{
        code = {
            \fill[black!10, rounded corners = 2pt] (-0.07, -0.07) rectangle (#1 + 0.07, #2 + 0.07);
            \draw[fill = white, thick] (0, 0) rectangle (#1, #2);
        }
    },
    pics/comm-rect/.style n args = {4}{
        code = {
            \draw[
                thick, #3, fill = #3!50, fill opacity = 0.2, rounded corners = 3pt]
                (-#1, -#2) rectangle (#1, #2);
            \node[draw, fill, circle, inner sep = 0pt, minimum size = 0.05cm, #3] (#4) at (0, 0) {};
        }
    },
    fancy-arrow/.style = {
        draw = black,
        thick,
        single arrow,
        right color = VioletRed!50,
        left color = ForestGreen!50,
        fill opacity = 0.7,
        single arrow head extend = 0.3cm,
        single arrow tip angle = 70,
        single arrow head indent = 0.15cm,
        minimum height = 2.4cm,
        minimum width = 0.8cm,
        shading angle = -45,
        rotate = -90
    }
}

%% file: pics/cube-decision.tex
\begin{tikzpicture}
    \def\a{3}
    \def\b{4}
    \def\rat{\b / \a}

    \def\colorlist{{
        "VioletRed", "ForestGreen", "Periwinkle", "Peach"
    }}

    \fill[black!10, rounded corners = 2pt] (-0.1 - \a, 0) -- (0, {-\b - 0.1 * \rat}) --
        (\a + 0.1, 0) -- (0, {\b + 0.1 * \rat}) -- cycle;
    \draw[very thick, fill = white, name path = fr] (-\a, 0) -- (0, -\b) -- (\a, 0) -- (0, \b) -- 
        cycle;
    \draw[dashed, thin] (-\a, 0) -- ++(2 * \a, 0);
    \node[below = 3pt] at (0.08, -\b) {$0^n$};
    \node[above = 3pt] at (0.08, \b) {$1^n$};

    \foreach \i in {1, 2}{
        \path[name path = val-\i] (-\a, {(0.35 * (2 * \i - 3) * \b)}) -- ++(2 * \a, 0);
        \path[name intersections = {of = fr and val-\i}];
        \draw[thick] (intersection-1) -- (intersection-2);
    }

    \pic {rhom = {{(\a - 0.07)}}{{(\b - 0.07)}}{VioletRed}{a1}};
    \pic[shift = {(0.5, 0.5)}] {rhom = {2}{\rat * 2}{ForestGreen}{a2}};
    \pic[shift = {(1, 1.3)}] {rhom = {0.85}{\rat * 0.85}{Periwinkle}{a3}};
    \pic[shift = {(1, 1.7)}] {rhom = {0.5}{\rat * 0.5}{Peach!70!black}{a4}};
    
    \draw[->, linecol = {0.3pt}{VioletRed}{ForestGreen}] (a1) -- (a2);
    \draw[->, linecol = {0.3pt}{ForestGreen}{Periwinkle}] (a2) -- (a3);
    \draw[->, linecol = {0.3pt}{Periwinkle}{Peach!70!black}] (a3) -- (a4);

    \node at (0, 0.75 * \b) {\Large $1$};
    \node at (0, -0.75 * \b) {\Large $0$};

    \node at (0, 4.7) {};
    \node at (0, -4.7) {};
\end{tikzpicture}

%% file: pics/assign-decision.tex
\begin{tikzpicture}
    \def\a{6}
    \def\b{2}
    \def\rat{\b / \a}

    \def\colorlist{{
        "VioletRed", "ForestGreen", "Periwinkle", "Peach"
    }}
    \def\vallist{{
        "0", "1", "1", "1", "1", "0", "1", "0", "1"
    }}

    \tikzmath{
        int \im;
        \hm = 0;
        \size = \b;
        for \im in {0, 1, ..., 3}{
            {
                \begin{scope}[shift = {(0, \hm)}]
                    \pgfmathparse{\colorlist[\im]}
                    \edef\col{\pgfmathresult}
                    \pic[\col] {universe-rect = {\a}{\size}};
                    \node[left, \col] at (0, \size / 2) {$X$};
                    \draw[\col, step = 0.2 * \b, thin] (0, 0) grid (0.6 * \b * \im, \size);
                    \pgfmathtruncatemacro{\t}{3 * \im - 1}
                    \ifthenelse{
                        \t > 0
                    }{
                        \foreach \i in {0, 1, ..., \t}{
                            \pgfmathparse{\vallist[\i]}
                            \edef\val{\pgfmathresult}
                            \pgfmathtruncatemacro{\temp}{4 - \im}
                            \foreach \j in {0, 1, ..., \temp}{
                                \node[\col] at ({(0.5 + \i) * 0.2 * \b}, {(0.5 + \j) * 0.2 * \b})
                                {$\val$};
                            }
                        }
                    }{};
                    \fill[
                        pattern = {Lines[angle = 45, distance = 3pt, line width = 0.1pt]},
                        pattern color = blue]
                        (0.6 * \b * \im, 0) rectangle (\a, \size);
                    \draw[\col,
                        decoration = {brace, mirror, raise = 2pt, amplitude = 3pt},
                        decorate] (0.6 * \b * \im, 0) -- (\a, 0) node[midway, below = 5pt] {\small $F$};

                \end{scope}
            };
            \size = \size - 0.2 * \b;
            \hm = \hm - \size - 0.8;
        };
    }

    \draw[pattern = {Lines[angle = 45, distance = 3pt, line width = 0.1pt]}, pattern color = blue]
        (0, 0) rectangle (\a, \b);
    \node[above] at (\a / 2, \b) {$[n]$};

    \node at (0, 2.7) {};
    \node at (0, -6.7) {};
\end{tikzpicture}

%% file: pics/cube-comm.tex
\begin{tikzpicture}
    \def\a{3}
    \def\b{4}
    \def\rat{\b / \a}

    \def\colorlist{{
        "VioletRed", "ForestGreen", "Periwinkle", "Peach"
    }}

    \pic[shift = {(-\a, -\a)}] {universe-rect = {2 * \a}{2 * \a}};
    \node[above] at (0, \a) {$\Sigma^n$};
    \node[right] at (\a, 0) {$\Sigma^n$};

    \pic {comm-rect = {{(\a - 0.06)}}{{(\a - 0.06)}}{VioletRed}{r1}};
    \pic[shift = {(0.5, 0.5)}] {comm-rect = {2}{2}{ForestGreen}{r2}};
    \pic[shift = {(1.6, 1.3)}] {comm-rect = {0.85}{\rat * 0.85}{Periwinkle}{r3}};
    \pic[shift = {(1.3, 1.4)}] {comm-rect = {0.5}{\rat * 0.5}{Peach!70!black}{r4}};

    \draw[->, linecol = {0.3pt}{VioletRed}{ForestGreen}] (r1) -- (r2);
    \draw[->, decorate,
        decoration = {snake, segment length = 1mm, amplitude = 0.3mm, post length = 1mm}] (r2) -- (r3);
    \draw[->, linecol = {0.3pt}{Periwinkle}{Peach!70!black}] (r3) -- (r4);

    \node[fancy-arrow, scale = 0.5] (ar) at (0, -\a - 1.4) {};
    \node[right = 5pt] at (ar) {$\IP^n$}; 
    
    \begin{scope}[shift = {(0, -2 * \a - 4)}]
        \fill[black!10, rounded corners = 2pt] (-0.1 - \a, 0) -- (0, {-\b - 0.1 * \rat}) --
            (\a + 0.1, 0) -- (0, {\b + 0.1 * \rat}) -- cycle;
        \draw[very thick, fill = white, name path = fr] (-\a, 0) -- (0, -\b) -- (\a, 0) -- (0, \b) --
            cycle;
        \draw[dashed, thin] (-\a, 0) -- ++(2 * \a, 0);
        \node[below = 3pt] at (0.08, -\b) {$0^n$};
        \node[above = 3pt] at (0.08, \b) {$1^n$};

        \foreach \i in {1, 2}{
            \path[name path = val-\i] (-\a, {(0.35 * (2 * \i - 3) * \b)}) -- ++(2 * \a, 0);
            \path[name intersections = {of = fr and val-\i}];
            \draw[thick] (intersection-1) -- (intersection-2);
        }

        \pic {rhom = {{(\a - 0.07)}}{{(\b - 0.07)}}{VioletRed}{a1}};
        \pic[shift = {(0.4, 1.2)}] {rhom = {1.5}{\rat * 1.5}{ForestGreen}{a2}};
        \pic[shift = {(0.9, 1.2)}] {rhom = {0.85}{\rat * 0.85}{Periwinkle}{a3}};
        \pic[shift = {(0.9, 1.6)}] {rhom = {0.5}{\rat * 0.5}{Peach!70!black}{a4}};

        \draw[->, linecol = {0.3pt}{VioletRed}{ForestGreen}] (a1) -- (a2);
        \draw[->, decorate,
            decoration = {snake, segment length = 1mm, amplitude = 0.3mm, post length = 1mm}] (a2) -- (a3);
        \draw[->, linecol = {0.3pt}{Periwinkle}{Peach!70!black}] (a3) -- (a4);

        \node at (0, 0.8 * \b) {\Large $1$};
        \node at (0, -0.8 * \b) {\Large $0$};

        \node at (0, 4.7) {};
        \node at (0, -4.7) {};        
    \end{scope}

\end{tikzpicture}

%% file: pics/assign-comm.tex
\begin{tikzpicture}
    \def\a{6}
    \def\b{2}
    \def\rat{\b / \a}
    
    \def\colorlist{{
        "VioletRed", "ForestGreen", "Periwinkle", "Peach"
    }}
    \def\vallist{{
        "1", "1", "0", "0", "1", "1", "1", "0", "1"
    }}
    \def\shi{{
        "2", "3", "2", "2"
    }}

    \tikzmath{
        int \im;
        \gs = 0;
        \hm = 0;
        \size = \b;
        for \im in {0, 1, ..., 3}{
            {
            \begin{scope}[shift = {(0, -\a + \hm - 0.25)}]
                \pgfmathparse{\colorlist[\im]}
                \edef\col{\pgfmathresult}
                \foreach \j in {0, 1}{
                    \begin{scope}[shift = {(0, {\j * (\size / 3 + 0.1)})}]
                        \pic[\col] {universe-rect = {\a}{\size / 4}};
                        \draw[\col, step = 0.2 * \b, thin] (0.2 * \b * \gs, 0) -- ++(0, \size / 4);

                        \fill[
                            pattern = {Lines[angle = 45, distance = 3pt, line width = 0.1pt]},
                            pattern color = blue] (0.2 * \b * \gs, 0) rectangle (\a, \size / 4);
                    \end{scope}
                }
                \node[left, \col] at (0, \size / 8) {$X$};
                \node[left, \col] at (0, \size / 3 + 0.1 + \size / 8) {$Y$};

                \draw[\col,
                    decoration = {brace, mirror, raise = 2pt, amplitude = 3pt},
                    decorate] (0.2 * \b * \gs, 0) -- (\a, 0) node[midway, below = 3pt] {\small spread};
                \ifthenelse{
                    \im > 0
                }{
                    \draw[\col,
                        decoration = {brace, mirror, raise = 2pt, amplitude = 3pt},
                        decorate] (0, 0) -- (0.2 * \b * \gs, 0) node[midway, below = 3pt] {\small fixed};
                }{};
                \end{scope}
            };
            \gs = \gs + \shi[\im];
            \size = \size - 0.2 * \b;
            \hm = \hm - \size * 0.65 - 0.8;
        };
    }
    \node[above] at (\a / 2, -\a + 1) {$[n]$};

    \node[fancy-arrow, scale = 0.5] (ar) at (\a / 2, -2 * \a - 0.4) {};
    \node[right = 5pt] at (ar) {$\IP^n$}; 
    
    \begin{scope}[shift = {(0, -2 * \a - 4)}]
        \tikzmath{
            int \im;
            \gs = 0;
            \hm = 0;
            \size = \b;
            for \im in {0, 1, ..., 3}{
                {
                \begin{scope}[shift = {(0, \hm)}]
                    \pgfmathparse{\colorlist[\im]}
                    \edef\col{\pgfmathresult}
                    \pic[\col] {universe-rect = {\a}{\size}};
                    \node[left, \col] at (0, \size / 2) {$\IP(X \times Y)$};
                    \draw[\col, step = 0.2 * \b, thin] (0, 0) grid (0.2 * \b * \gs, \size);
                    \pgfmathtruncatemacro{\t}{\gs - 1}
                    \ifthenelse{
                        \t > 0
                    }{
                        \foreach \i in {0, 1, ..., \t}{
                            \pgfmathparse{\vallist[\i]}
                            \edef\val{\pgfmathresult}
                            \pgfmathtruncatemacro{\temp}{4 - \im}
                            \foreach \j in {0, 1, ..., \temp}{
                                \node[\col] at ({(0.5 + \i) * 0.2 * \b}, {(0.5 + \j) * 0.2 * \b})
                                {$\val$};
                            }
                        }
                    }{};
                    \fill[
                        pattern = {Lines[angle = 45, distance = 3pt, line width = 0.1pt]},
                        pattern color = blue] (0.2 * \b * \gs, 0) rectangle (\a, \size);
                    \draw[\col,
                        decoration = {brace, mirror, raise = 2pt, amplitude = 3pt},
                        decorate] (0.2 * \b * \gs, 0) -- (\a, 0) node[midway, below = 5pt] {\small $F$};
                    \ifthenelse{
                        \im = 2
                    }{
                        \draw[decoration = {brace, mirror, raise = 2pt, amplitude = 3pt},
                            decorate] (0.2 * \b * 2, 0) -- ++(0.2 * \b * 3, 0)
                            node[midway, below = 5pt] {\small unsafe};
                    }{};

                \end{scope}
                };
                \gs = \gs + \shi[\im];
                \size = \size - 0.2 * \b;
                \hm = \hm - \size - 0.8;
            };
        }

        \draw[pattern = {Lines[angle = 45, distance = 3pt, line width = 0.1pt]}, pattern color = blue]
            (0, 0) rectangle (\a, \b);
        \node[above] at (\a / 2, \b) {$[n]$};
        \node at (0, 2.7) {};
        \node at (0, -6.7) {};
    \end{scope}

\end{tikzpicture}